\documentclass[11pt,a4paper]{article}
\usepackage{hyperref}
\usepackage{pgf}
\usepackage{pgflibraryarrows}
\usepackage{pgflibrarysnakes}
\usepackage[utf8]{inputenc}
\usepackage[english]{babel}
\usepackage[normalem]{ulem}
\usepackage{graphicx}
\usepackage{amsmath}
\usepackage{amssymb,mathrsfs}
\usepackage{amsthm}
\usepackage{cite}
\usepackage{comment}
\usepackage{tikz}
\usepackage{authblk}
\usepackage{todonotes}
\usepackage{cancel}
\usepackage{caption}
\usepackage{subcaption}
\usepackage{geometry} 

\newtheorem{theorem}{Theorem}[section]
\newtheorem{lemma}{Lemma}[section]

\newtheorem{remark}{Remark}[section]
\newtheorem{problem}{Problem}[section]
\newtheorem{assumption}[lemma]{Assumption}

\numberwithin{equation}{section}
\DeclareMathOperator{\dn}{dn}

\newcommand{\C}{\mathbb{C}}
\newcommand{\R}{\mathbb{R}}
\newcommand{\D}{\mathcal{D}}
\newcommand{\I}{\mathcal{I}}

\def\d{{\rm d}}
\def \ov{\overline}
\def \e{{\rm e}}
\def \im{\, {\rm Im}}
\def \re{\, {\rm Re}}
\def \nn{\nonumber}
\def\id{\mathrm {Id}\,}
\def \pa{\partial}
\def\1{ \mathbb{I}}

\def\bea#1\eea{\begin{align}#1\end{align}}
\def\be#1\ee{\begin{align}#1\end{align}}

\newcommand\EN           {\end{equation}}
\newcommand\bes           {\begin{subequations}}
\newcommand\esu           {\end{subequations}}

\def\3pt#1#2#3{{\langle{#1}\vert{#2}\vert{#3}\rangle}}

\title{$\bar{\partial}$-problem for  focusing  nonlinear Schr\"odinger equation   and soliton shielding}

\author[1]{Marco Bertola}
\affil[1]{Concordia University, 1455 av. de Maisonneuve W.,  Montr\'eal Canada}
\affil[2]{SISSA, via Bonomea 265, 34136, Trieste, Italy  and INFN sezione di Trieste}

\author[2,3]{Tamara Grava}
\affil[3]{School of Mathematics, University of Bristol, Fry Building, Bristol,
BS8 1UG, UK}
\author[4]{Giuseppe Orsatti}
\affil[4]{Institute de Recherche en Mathématique et Physique, UCLouvain, Chemin du Cycloton 2, 1348 Louvain-la-Neuve, Belgium}

\begin{document}

\maketitle
\begin{abstract}
We consider  soliton gas solutions of the Focusing Nonlinear Schr\"odinger (NLS) equation, where the point spectrum of the Zakharov-Shabat linear operator condensate in a bounded domain $\D$ in the upper half-plane.  We show that the  corresponding  inverse  scattering  problem can be formulated as a $\ov \pa$-problem on the domain.  We prove the existence of the solution of this  $\ov \pa$-problem  by showing that  the  $\tau$-function of the problem (a Fredholm determinant) does not vanish. We  then represent the solution of the NLS equation via the $\tau$ of the $\ov\partial$- problem.   Finally we show  that, when the domain $\D$ is an ellipse and the density of solitons is analytic,  the initial datum of the Cauchy problem is asymptotically step-like oscillatory,  and it is  described by a  periodic elliptic function as $x \to - \infty$ while it  vanishes exponentially fast as $x \to +\infty$. 
\end{abstract}

\section{Introduction}
In this manuscript we consider  the focusing nonlinear Schr\"odinger equation  (NLS)
\begin{equation}
\label{NLS}
i\psi_t+\frac{1}{2}\psi_{xx}+|\psi|^2\psi=0,
\end{equation}
with nonstandard initial  data that originate from  the infinite soliton limit.  
The NLS equation is an integrable equation and the Cauchy problem  can be solved via the inverse scattering transform of the Zakharov-Shabat linear operator \cite{ZS1971}  when the initial data satisfies  zero boundary conditions  at infinity, it is periodic \cite{Faddeev2007} or  the initial data  has asymmetric boundary conditions  \cite{Biondini2014, Demontis}, or  also when one considers an initial boundary value problem instead of a Cauchy problem (see e.g. \cite{FIS2005},\cite{IS2013}).
In all these cases the inverse problem can be formulated as a Riemann-Hilbert Problem {\color{black} for a $2\times 2$  matrix function defined on the complex plane with discontinuities along paths and with a certain number of  poles. These paths and poles  correspond to the support of the   spectral scattering  data. 
Here   we consider instead the case in which the inverse problem is  necessarily  cast    as a  $\overline{\partial}$-problem because the spectral scattering data has a two-dimensional support.}
{\color{black}
The  $\overline{\partial}$-method  in inverse scattering was developed by Fokas-Ablowitz \cite{FA1,FB2}  and Beals-Coifman \cite{BC1,BC2} to study the Cauchy problem of integrable, dispersive nonlinear equations in two or more space dimensions. Its application to
inverse scattering has been studied by many authors  (see  for example 
the monograph \cite{AblowitzClarkson} for references up to 1990, and Grinevich, Grinevich-Novikov  \cite{Grin1,GrinNovi,GrinNovi2}).

 To formulate the  $\overline{\partial}$-problem for the NLS equation  let  $\mathcal D$ denote a compact domain with smooth boundary in the upper half--plane $\mathcal{D} \subseteq \mathbb{C}_+$  and  let  $\beta:\mathcal{D}\to\C$ be a smooth bounded function continuous up to the boundary of $\mathcal{D}$,  that describes the spectral data (soliton density) of the problem. The $\overline{\partial}$-problem is to find a $2\times 2$}  matrix function $\Gamma(z; x,t)$  depending on the complex variable $z\in\C$ and $(x,t)\in\R\times\R^+$ that  satisfies the conditions
\begin{equation}
  \label{eq:d_bar}
  \left\{
  \begin{split}
    \bar{\partial}\Gamma(z;x,t)&=\Gamma(z; x,t)M(z;x,t) \\
    \Gamma(z;x,t) &= \mathbb{I}+\mathcal{O}(z^{-1}),\quad \mbox{as $z\to\infty$}. 
  \end{split} \right.
\end{equation}
Here 
\begin{equation}
\label{M_intro}
  M(z;x,t):=\left\{
  \begin{array}{ll}
    \begin{pmatrix}
      0 & -\beta^{*}(z)^2e^{-2\theta(z;x,t)}\\
      0 & 0
    \end{pmatrix}  \text{ for } z \in\ov \D  \\
    \begin{pmatrix}
      0 & 0\\
      \beta(z)^2 e^{2\theta(z;x,t)}& 0
    \end{pmatrix}  \text{ for } z \in \mathcal{{D}},
  \end{array}\right.
\end{equation}
\begin{equation}
\theta(z;x,t)= i(z^{2}t + z x),
\end{equation}
where  $\beta^{*}(z):=\overline{\beta(\bar{z})}$   and $\ov \D=\{z\in\C\, |\,\overline{z}\in\D\}$.

If the solution $\Gamma(z;x,t)$ of the  $\ov\partial$-problem  \eqref{eq:d_bar}  exists, the solution $\psi$,  of the NLS equation,  is recovered from the solution  $\Gamma(z;x,t)$ by the formula
\begin{equation}
\label{Nsoliton}
\psi(x,t)=2i \lim_{z \to \infty}z(\Gamma(z;x,t))_{12}.
\end{equation}
The $\ov\partial$-problem for non zero boundary condition has  recently  been considered in \cite{Zhu2023}.
The $\ov\partial$-problem \eqref{eq:d_bar} was formulated in \cite{BGO} in the setting of integrable operators (see section below).
However the existence of solutions of such family of $\ov\partial$-problems has remained an open problem.
Furthermore, for general $\D$ and smooth $\beta$ the   class of initial data $\psi(x,0)$  described by the $\ov\partial$-problem \eqref{eq:d_bar}
is in general unknown and requires the asymptotic analysis of $\ov\partial$-problems  {\color{black}  and the   extension of the  $\ov\partial$-steepest descent method  developed by   McLaughlin and  Miller and Dieng
\cite{McLaughlin2008},\cite{Dieng2019}.}
If $\beta $ is analytic in $\D$ and the domain is a ``generalized quadrature domain''  \cite{Gustafsson}, the $\ov\partial$-problem \eqref{eq:d_bar} can be reduced to a classical Riemann-Hilbert problem with discontinuities along a  collection 
$\mathcal{L}$ of  arcs. We call this reduction {\it soliton shielding}   because the effective soliton charge can be reduced from a two-dimensional domain to a collection of arcs.

%
An example of this case, first spotted in \cite{BGO},  is the ellipse where  $\mathcal L$ is the segment joining the foci. 
The goals of this manuscript are:
\begin{itemize}
\item     {\color{black}  to show the existence of the   solution of the  $\ov \partial$-problem  \eqref{M_intro} and to show that such solution is obtained as an infinite soliton limit  (see Theorem~\ref{th:FD} and  formula \eqref{psi_fredholm}); on the way we represent   $|\psi(x,t)|^2$ by the  second $x~$ derivative of the logarithmic of the $\tau$-function of the  $\ov \partial$-problem;}
\item    to show that  when  $\mathcal{D}$ is an ellipse  and $\beta(z)$  is analytic,  the solution of the $\ov \partial$-problem   \eqref{M_intro} produces an   initial datum  that is step-like oscillatory. 
In particular we show  in Theorem~\ref{thmPsi0} that  
\begin{equation}
\label{psiointro}
 \psi(x,0)=\left\{
 \begin{array}{ll}
 & -i \e^{2i(g_{\infty}x+\phi_\infty)}(\alpha_2 +\alpha_1) 
		\dn \left( (\alpha_2 +\alpha_1)  ( x -x_0); m \right) +\mathcal{O}(\e^{c_- x}),\quad \mbox{as $x\to-\infty$},\\
		&\mathcal{O}(\e^{-c_+ x}),\quad \mbox{as $x\to+\infty$}
		\end{array}\right.
\end{equation}
where  $c_\pm$ are positive constants, $\dn(z;m)$ is the Jacobi elliptic function of modulus $m=\frac{4 \alpha_2 \alpha_1}{(\alpha_2+ \alpha_1)^2}$  and $x_0$ 
 is a constant that depends on   $ \beta(z)$ and the geometry of the problem, $g_{\infty}$ and $\phi_{\infty}$ are  real  constants, and $\alpha_1$ and $\alpha_2$ are the  location of the  foci of the ellipse $\mathcal{D}$  on the positive imaginary axis.
\end{itemize}
This manuscript is organized as follows. 
In Section~\ref{sec2} we show how to obtain the $\ov \partial$-problem  \eqref{M_intro} from a $N$-soliton  solution of the  NLS equation in the  limit $N\to\infty$. In Section~\ref{sec3} we  prove solvability of the  $\ov \partial$-problem  \eqref{M_intro} by showing the non-vanishing of its  $\tau$-function.
In Section~\ref{sec4} we   give an asymptotic characterisation of the initial data $\psi(x,0)$ when $\mathcal{D}$  is an ellipse and the function $\beta$  is analytic   in $\mathcal{D}$. In particular we show that such initial data   is step-like oscillatory as described by \eqref{psiointro}. Such an initial datum $\psi(x,0)$,   despite originating from a soliton spectra uniformly covering a two-dimensional domain,  is similar to the initial datum obtained  for the Korteweg de Vries equation and the modified Korteweg de Vries equation in  \cite{Girotti2021,GGJMM} where the soliton spectra fill  uniformly a segment of the imaginary axis of the  spectral plane. 

\section{Infinite soliton limit and $\overline{\partial}$-problem}
\label{sec2}
The $\overline{\partial}$-problem \eqref{eq:d_bar} appears when  considering the  limit of  infinitely many solitons  of  the NLS initial data.
 The one--soliton solution is given by 
\begin{equation}
\label{one_soliton}
\psi(x,t)=2b\,\mbox{sech} [2b(x+2at-x_0)]{\rm e}^{-2i[ax+(a^2-b^2)t+\frac{\phi_0}{2}]}
\end{equation}
where $x_0$ is the initial peak position of the soliton, $\phi_0$ is the initial phase, $2b$ is the modulus of the wave maximal amplitude and $-2a$ is the soliton velocity.  
The $N$ soliton solution can be obtained from the inverse scattering problem for the   Zhakarov-Shabat linear spectral problem with only   discrete spectrum $\{z_j\}_{j=1}^N$
and norming constants $\{c_j\}_{j=1}^N\in\C^N\setminus\{0\}$.
 When the solitons are far apart,  each  spectral point   $z_j=a_j+ib_j$   describes a soliton with  speed $-2a_j\in\R$  and amplitude $2b_j>0$.
  The phase space of the $N$-soliton solution is $\C_+^N\times \{\C^N\setminus\{0\}\}$ where $\C_+$ is the upper half space.

 The inverse scattering problem   recovers   the $N$-soliton solution from the spectral data  $\{z_j,c_j\}_{j=1}^N$.
 This goal is accomplished by formulating the inverse scattering problem  as a Riemann-Hilbert Problem  for a $2\times 2$  matrix $Y^N(z;x,t)$   where $z\in \C$  and $(x,t)\in\R\times\R^+$.  Here we follows the presentation of the $N$-soliton solution implemented in \cite{Girotti2021},\cite{GGJMM}, to obtain the infinite soliton limit.
 
Suppose that the points $z_j=a_j+ib_j$ are contained within a counterclockwise close  curve $\gamma_+$ in the upper half space and similarly the points $\overline{z}_j=a_j-ib_j$ are contained in  $\gamma_-=\{z\in\C \;|\; \overline{z}\in\gamma_+\}$ that is  oriented  counterclockwise.
The matrix $Y^N(z;x,t)$ is required to be  analytic for $z\in\C\backslash\{\gamma_+\cup\gamma_-\}$ and 
 with boundary values satisfying
\begin{equation}
  \label{eq:new_R-H0}
  \begin{split}
&Y^N(z_{+};x,t)=Y^N(z_-;x,t)
  \begin{pmatrix}
    1 &  
    \sum\limits^{N}_{j=1}\frac{\bar{c}_j e^{-2\theta(z,x,t)}}{z-\bar{z}_j}{\bf 1}_{\gamma_-}(z)\\
     - 
     \sum\limits_{j=1}^{N}\frac{c_j e^{2\theta(z,x,t)}}{z-z_j}{\bf 1}_{\gamma_+}(z) & 1
  \end{pmatrix},\quad 
   \end{split}
\end{equation}
for $z\in \gamma_+\cup\gamma_-$ where ${\bf 1}_{\gamma_{\pm}}$ are the characteristic function of the curves $\gamma_{\pm}$. 
The normalization of the problem is fixed by the  behaviour at infinity
\begin{equation}
  \label{eq:new_R-H1}
  \begin{split}
  &Y^N(z;x,t)\underset{z \to \infty}{ \longrightarrow } \1\,.
  \end{split}
\end{equation}
From the inverse scattering problem one finds that 
\begin{equation}\label{Yexpand}
	Y^N(z;x,t) = I + \frac{1}{2i z} 
		\begin{pmatrix} 
		- \int_{x}^\infty |\psi(s,t)|^2 ds &
		\psi(x,t) \\
		\psi(x,t)^* &
		 \int_{x}^\infty |\psi(s,t)|^2 ds 
		\end{pmatrix}
		+ O(z^{-1}).
\end{equation}
From \eqref{eq:new_R-H0}  and \eqref{eq:new_R-H1}  the matrix $Y^N(z;x,t)$  can be reconstructed by the formula

\begin{equation}
\label{solution}
Y^N(z;x,t)= \left[\mathbb{I}+\sum_{j=0}^{N-1}\dfrac{\begin{pmatrix}l_j(x,t)&0\\
m_j(x,t)&0\end{pmatrix}}{z-z_j}+\sum_{j=0}^{N-1}\dfrac{\begin{pmatrix}0&-\overline{m_j(x,t)}\\
0&\overline{l_j(x,t)}\end{pmatrix}}{z-\overline{z_j}}\right]J_Y(z)\,,
\end{equation}
where 
\[
J_Y(z)=  \begin{pmatrix}
    1 & 
    
    \sum\limits^{N}_{j=1}\frac{\bar{c}_j e^{-2\theta(z,x,t)}}{z-\bar{z}_j}{\bf 1}_{D_{\gamma_-}}(z)\\
     -
     
     \sum\limits_{j=1}^{N}\frac{c_j e^{2\theta(z,x,t)}}{z-z_j}{\bf 1}_{D_{\gamma_+}}(z) & 1
  \end{pmatrix},\quad 
\]
and $D_{\gamma_\pm}$ are the domains bounded by $\gamma_\pm$, respectively. The  request that $Y^N(z;x,t)$ is analytic  for $z\in\C\backslash\{\gamma_+\cup\gamma_-\}$  leads to a determined  linear system of equations for  the coefficients $l_j$ and $m_j$. 
Then, the reconstruction formula of the  $N$-soliton solution  $\psi_N(x,t)$ is obtained  from \eqref{Yexpand} by 
\[
\psi_N(x,t)=2i \lim_{z \to \infty}z(Y^N(z;x,t))_{12}=-2i\sum_{j=1}^N\overline{m_j(x,t)}.
\]
A direct expression of the amplitude of the solution can be obtained  from \eqref{Yexpand} using the Kay-Moses formula \cite{Kay1955},
\begin{equation}
  \label{eq:K-M}
  |\psi_{N}(x,t)|^2 = 2i\partial_x\sum_{j=1}^Nl_j(x,t)=\pa^2_{x}\log({\rm det}(\1_{\C^N} + \Phi_{N}(x,t)\,\ov{\Phi_{N}(x,t)})), 
\end{equation}
 where $\Phi_N$ is an $N \times N$ matrix with elements
\begin{equation}
\label{Phi}
  (\Phi_{N}(x,t))_{jk}:= \frac{\sqrt{c_j\ov{c_k}}e^{\theta(z_j,x,t)-\theta(\ov{z_k},x,t)}}{i(z_j-\ov{z_k})}, \quad \Phi_{N}(x,t)^{T}=\ov{\Phi_{N}(x,t)}. 
\end{equation}
We observe that $\Phi_N$  is a Hermitian matrix. The  determinant expression~\eqref{eq:K-M} was derived in~\cite{Borghese2018}.

We  consider the limit  $N \to \infty$ as  in \cite{BGO}.  Let  $\mathcal{D}$  be the simply connected domain bounded by $\gamma_+$ (and $\ov \D $ bounded by  $\gamma_-$).  
We assume that the poles $z_j$  accumulate uniformly in $\mathcal{D}$ as $N$ grows.  We  choose the norming constants $c_j$ to be interpolated by a smooth function $\beta^2 (z, \ov z)$ in $\mathcal{D}$ (we write for brevity $\beta^2(z)$ where it is understood that the function is smooth and not necessarily analytic)
\begin{equation}
  \label{eq:norm-const}
c_j=\frac{\mathcal{A}}{\pi N}\beta^2(z_j){\bf 1}_{\mathcal{D}}(z_j),
\end{equation}
where ${\bf 1}_{\mathcal{D}}$ is the characteristic function of the domain $\mathcal{D}$  and  $\mathcal{A}$ is the area  of $\mathcal{D}$.
Upon taking the limit $N \to +\infty$, we get
\begin{equation}
  \sum_{j=1}^{N}\frac{c_j}{(z-z_j)}=\sum_{j=1}^{N}\frac{\mathcal{A}}{\pi N}\frac{\beta^2(z_j)}{(z-z_j)} \xrightarrow[N\to\infty]{} \iint_{\mathcal{D}}\frac{\beta^2(w)}{z-w}\frac{\d^2w}{\pi}
\end{equation}
where $\d^2w = \frac{\d\bar{w} \wedge \d w}{2 i}=\d x\d y$ is the usual area element. 

Consequently, the RHP~\eqref{eq:new_R-H0},\eqref{eq:new_R-H1},  for $(x,t)$ in compact sets  in the limit $N\to\infty$,  is the following.
\begin{problem}
  We are looking for a $2\times 2$ matrix  $Y^\infty(z;x,t)$  analytic in $\C\backslash \{\gamma_+\cup\gamma_-\}$ such that 
\begin{equation}
  \label{eq:new_R-H_2}
  \begin{split}
    Y^\infty_{+}(z;x,t)&=Y^\infty_{-}(z;x,t)J_\infty(z;x,t),\\
    Y^\infty_{+}(z;x,t)&= \1 + \mathcal{O}\left(z^{-1}\right), \text{ as } z \to \infty
    \end{split}
\end{equation}
with jump matrix
\begin{equation}
\label{Jinfty}
  J_\infty(z;x,t)=
   \begin{pmatrix}
    1 & \iint\limits_{\overline{\mathcal{D}}}\frac{e^{-2\theta(z,x,t)} \beta^{*}(w)^2\d^2w}{\pi(z-w)}{\bf 1}_{\gamma_{-}}(z)\\
   - \iint\limits_{\mathcal{D}}\frac{e^{2\theta(z,x,t)} \beta(w)^2\d^2w}{\pi(z-w)}{\bf 1}_{\gamma_{+}}(z)& 1
  \end{pmatrix},
\end{equation}
where we have introduced the notation $\beta^*(z) := \ov {\beta(\ov z)}$.
\end{problem}
We can relate this RHP to a $\ov\pa$-problem as we now explain. 
Consider  the transformation
\begin{equation}
  \label{eq:trans_2}
 \Gamma(z;x,t)=Y^\infty(z;x,t)A(z;x,t)
\end{equation}
where the matrix $A(z;x,t)$ is defined by 
\begin{equation*}
  A(z;x,t)=\left\{
    \begin{array}{ll}
     & \begin{pmatrix}
       1 & -\iint\limits_{\ov \D }\frac{e^{-2\theta(z,x,t)} \beta^{*}(w)^2\d^2w}{\pi(z-w)}\\
       0 & 1
     \end{pmatrix} \text{ for } z \text{ inside the loop } \gamma_{-}\\
    & \begin{pmatrix}
       1 & 0\\
        \iint\limits_{\mathcal{D}}\frac{e^{2\theta(z,x,t)} \beta(w)^2\d^2w}{\pi(z-w)} & 1
      \end{pmatrix} \text{ for } z \text{ inside the loop } \gamma_{+}\\
    & \1 \text{ otherwise }  
    \end{array} \right.
\end{equation*}
We apply the $\ov{\partial}$ operator to~\eqref{eq:trans_2}
\begin{equation*}
  \ov{\partial}\Gamma(z;x,t)=Y^\infty(z;x,t)\ov{\partial}A(z;x,t)=\Gamma(z;x,t)M(z;x,t)
\end{equation*}
where \begin{equation}
\label{M}
  M(z;x,t)= A^{-1}(z;x,t)\ov{\partial}A(z;x,t)\,.
\end{equation}
A simple computation shows that 
 \begin{equation*}
 \begin{split}
  M(z;x,t)&= \begin{pmatrix}
    0 & -\beta^*(z)^2e^{-2\theta(z,x,t)}{\bf 1}_{\overline{\mathcal{D}}}(z)\\
    \beta(z)^2e^{2\theta(z,x,t)}{\bf 1}_{\mathcal{{D}}}(z) & 0
  \end{pmatrix},
\end{split}
\end{equation*}
so that  $\Gamma(z;x,t) $ defined in 
 \eqref{eq:trans_2} satisfies the $\overline{\partial}$-problem \eqref{eq:d_bar}.
 The NLS solution is still recovered by the formula
 \[
\psi(x,t)=2i \lim_{z \to \infty}z(\Gamma(z;x,t))_{12}.
\]
\begin{remark}
It would be interesting to consider a similar limit for higher order solitons  or even breathers  in the spirit of \cite{Bilman1}\cite{Bilman3}\cite{Bilman2019}.
\end{remark}

 \section{Fredholm determinant for the soliton gas}
 \label{sec3}
In~\cite{dbar2024},   it was shown that  the  $\overline{\partial}$-problems of the form \eqref{M} extend the theory of integrable operators of the   Its-Izergin-Korepin-Slavnov (IIKS) framework   \cite{IIKS}. The gist of \cite{dbar2024} is to consider  a Hilbert-Schmidt operator of  integrable type acting on $L^2(\mathscr{D},\d^2 w)\otimes \mathbb{C}^n$  to itself:
 \begin{equation}
   \label{eq:K-op}
   \mathcal{K}[v](z,\overline{z})=\iint_{\mathscr{D}}\hat{\mathcal{K}}(z,\overline{z},w,\overline{w})v(w)\d^2w , \; \; z \in \mathscr{D},
 \end{equation}
 with kernel of the form
 \begin{align}
 \label{integrableK}
   & \hat{\mathcal{K}}(z,\overline{z},w,\overline{w})= \frac{p(z,\overline{z})^{T}\,q(w,\overline{w})}{(z-w)}, & \\
   & p(z,\overline{z})^{T}\,q(z,\overline{z})\equiv 0 \equiv (\overline{\partial}p(z,\overline{z}))^{T}\,q(z,\overline{z}) & \qquad p,q \in \mathcal{C}^{\infty}(\mathscr{D}, {\rm Mat} (r \times n,\mathbb{C})).
 \end{align}
 Here $\mathcal D$ is just an arbitrary compact set  in the complex plane (consisting possibly of several connected components): as mentioned earlier, we will drop explicit notation of the $\ov z$ dependence when indicating a smooth function. 
The result of \cite{dbar2024} is that the resolvent operator  $\mathcal{R}$ of $\mathcal{K}$ exists if and only if the  following $\overline{\partial}$-problem  for the matrix  $\Gamma(z)\in GL_r(\mathbb C)$ admits a solution: 
\begin{equation}
   \label{eq:d-bar_IIKS}
   \overline{\partial}\Gamma(z)= \pi \Gamma(z) p(z)\,q(z)^{T}{\bf 1}_{\mathcal{D}}(z), \qquad \Gamma(z) \underset{z \to \infty}{ \to } \1.
 \end{equation}
 The resolvent operator  $\mathcal{R}$  has  kernel  $ \hat{\mathcal{R}}(z,w)$ expressed via $\Gamma$, $p$ and $q$  by the relation
 \begin{equation}
   \hat{\mathcal{R}}(z,w)= \frac{p(z)^{T}\Gamma^T(z)\Gamma^{-T}(w)q(w)}{z-w}\,.
 \end{equation}
In  \cite{dbar2024} we  also  showed  
that the  $\tau$-function of the $\ov \pa$-problem identified by 
the Hilbert--Carleman determinant  (see e.g. \cite{Gohberg_2000})
\begin{equation}
\label{tau}
\tau :=\text{det}_2\big[{\rm Id} - \mathcal K\big]= \text{{\rm det}}\left[(\id - \mathcal{K})\e^{\mathcal{K}}\right],
\end{equation}
defined  on the space $L^2(\mathscr{D}, \d^2 z)\otimes \mathbb{C}^n$.  The Hilbert--Carleman determinant is  well defined for Hilbert-Schmidt operators.  We then showed the following
\begin{theorem}\cite{dbar2024}
\label{Theorem3.1}
The operator $\id-\mathcal K$,  with $\mathcal K$ as in \eqref{eq:K-op} and kernel $ \hat{\mathcal{K}}(z,w)$ of the form \eqref{integrableK} is invertible in $L^2(\mathscr D, \d^2 z) \otimes \C^n$ if and only if the $\ov\pa$-problem \ref{eq:d-bar_IIKS} admits a solution.
\end{theorem}

The goal of this section is to identify the integrable operator $\mathcal{K}$ for the $\ov\partial$-problem \eqref{M} and show that the corresponding $\tau$-function is  positive, thus proving the solvability of the $\ov\partial$-problem \eqref{M}.

 In our case, we can factorize the matrix $M(z,\overline{z})$, defined in~\eqref{M}, in the matrix product of $2 \times 1$ vectors $p(z,\overline{z})$ and $q(z,\overline{z})$
 \begin{equation}
   p(z)= \frac{e^{-\theta(z,x,t)\sigma_3}}{\sqrt{\pi}}
   \begin{bmatrix}
     -{\beta^{*}(z){\bf 1}_{_{\ov \D}}(z)}\\
     {\beta(z){\bf 1}_{_ {\D}}(z)}
   \end{bmatrix}, \qquad 
   q(z,\overline{z})= \frac{{\rm e}^{\theta(z,x,t)\sigma_3}}{\sqrt{\pi}}
   \begin{bmatrix}
      \beta(z){\bf 1}_{_ {\D}}(z)\\
     \beta^{*}(z){\bf 1}_{_ {\ov \D}}(z)
   \end{bmatrix}.
 \end{equation}
From \eqref{integrableK}, the corresponding  operator $\mathcal{K}$   acting on  $L^{2}(\mathcal{D}\cup \ov \D)$  has  the following kernel:
 \begin{equation}
   \label{eq:Kern-NLS}
   \widehat{\mathcal{K}}(z,w)=\frac{{\beta(z)\beta^*(w)}{\rm e}^{\theta(z,x,t) - \theta(w,x,t)}{\bf 1}_{_{\mathcal{D}}}(z){\bf 1}_{_{\ov \D}}(w)}{\pi(z-w)} 
  -\frac{{\beta^{*}(z)\beta(w)}{\rm e}^{\theta(w,x,t) - \theta(z,x,t)}{\bf 1}_{_{\mathcal{D}}}(w)
  {\bf 1}_{_{\ov \D}}(z)}{\pi(z - w)}.
 \end{equation}
 
 The goal  is to show that 
$\text{det}_2\big[{\rm Id}_{L^2(\D\cup \ov \D)} - \mathcal K\big]$ is non vanishing for all $x,t\in \R$  thus proving that the $\ov\partial$-problem \eqref{eq:d_bar} has always a solution. 
 

 In order to proceed with our analysis, we  first revisit the theory of $N$-soliton solutions.
%
%

\paragraph{Review of the $N$-soliton solution.}

The $N$-soliton solution of the NLS equation~\eqref{NLS}  with scattering data $\{z_j,c_j\}_{j=1}^N$, $\im( z_j)>0$,  $c_j\in\C\backslash\{ 0 \}$  can be recovered by the Kay-Moses formula \eqref{eq:K-M}
 where  $\Phi$ is an $N \times N$ matrix with elements
\begin{equation*}
  (\Phi_{N}(x,t))_{jk}:= \frac{\sqrt{c_j\ov{c_k}}e^{\theta(z_j)-\theta(\ov{z_k})}}{i(z_j-\ov{z_k})}, \quad \Phi_{N}(x,t)^{T}=\ov{\Phi_{N}(x,t)}. 
\end{equation*}
We observe that $\Phi_N$  is a Hermitian matrix. 
The determinant appearing in the Kay-Moses formula  is precisely the finite--dimensional counterpart of the $\tau$-function and thus we set
\begin{equation}
  \label{eq:tau-N}
  \tau_{N}(x,t)= {\rm det}(\1_{\C^N} + \Phi_{N}(x,t)\,\ov{\Phi_{N}(x,t)}).
\end{equation}
{\color{black} This expression can be derived from the IIKS theory starting from the RHP \eqref{eq:new_R-H0} as we show in Appendix \ref{IIKStau} for the benefit of the reader.}
We can represent the matrix $\Phi_{N}$  as a composition of an operator with its adjoint: define $\mathcal{M}_{N}: L^2([x,+\infty)) \to \mathbb{C}^{N}$ and its adjoint $\mathcal{M}_{N}^{\dag}:\mathbb{C}^{N} \to  L^2([x,+\infty)) $ as
\begin{align} 
\label{defMMd}
  & \mathcal{M}_{N}[u]_{j}:=\int^{+\infty}_{x}\sqrt{c_j}e^{i(z_j s + z_j^2 t)} u(s)\d s , \\
  &\mathcal{M}^{\dag}_{N}[\vec{v}](y) := \sum_{k=1}^{N}\sqrt{\overline{c_k}}e^{-i (\ov{z_k} y +\ov{z_k}^2 t)}v_{k}, \ \ \ y\in [x,\infty).
\end{align}
A direct computation shows that  $\Phi_N=- \mathcal{M}_{N}\circ \mathcal{M}^{\dag}_{N}$. 
Therefore $\Phi_N$ is not only a Hermitian matrix, but a non-positive  definite one. Similarly the complex conjugate matrix $\ov{\Phi_N}$ (no transposition), is also a Hermitian, non-positive definite matrix and hence the determinant \eqref{eq:tau-N} is necessarily strictly positive.
%
We now consider the limit $N\to \infty$ and show a similar theorem of positivity for our Hilbert-Carleman determinant \cite{Gohberg_2000}.
\begin{theorem}
  \label{th:FD}
Let $\{z_j,c_j\}_{j=1}^{N}$ be the spectral data of the $N$ soliton solution.
Let us assume that $c_j$ are interpolated by a smooth function as in \eqref{eq:norm-const} and 
 the point spectrum $\{z_j\}_{j=1}^N$  accumulates uniformly,  as $N\to\infty$, on a domain $\mathcal{D}$.
 Then the function $\tau_{N}(x,t)$ defined in~\eqref{eq:tau-N} converges,  for $(x,t)$ in a compact set,  to  the $\tau$-function of the $\ov\partial$-problem \eqref{eq:d_bar}, namely  the  Hilbert-Carleman determinant
  \begin{equation}
    \label{eq:tau-inf}
    \tau(x,t)= {\rm det}_2(\id_{L^{2}(\mathcal{D}\cup \ov \D)} - \mathcal{K}),
  \end{equation}
  where $\mathcal{K}$ is a trace class integrable operator acting on  $L^2(\mathcal{D}\cup\ov \D)$ with kernel $\widehat{\mathcal{K}}$ given by
  \begin{align}
   & \label{eq:Kth} 
   \widehat{\mathcal{K}}(z,w)=\frac{{\beta(z)\beta^*(w)}{\rm e}^{\theta(z,x,t) - \theta(w,x,t)}{\bf 1}_{_{\mathcal{D}}}(z){\bf 1}_{_{\ov \D}}(w)}{\pi(z-w)} 
  -\frac{{\beta^{*}(z)\beta(w)}{\rm e}^{\theta(w,x,t) - \theta(z,x,t)}{\bf 1}_{_{\mathcal{D}}}(w)
  {\bf 1}_{_{\ov \D}}(z)}{\pi(z - w)}.
  \end{align}
 Moreover $\tau(x,t)>  0$ for all $x, t\in\R$ and therefore the solution of the $\ov\partial$-problem \eqref{eq:d_bar} exists.
\end{theorem}
\begin{remark}
The operator   $\mathcal{K}$ is a trace class operator (see \eqref{323} and the following discussion)  and $\mbox{Tr}\,\mathcal{K}=\int_{\mathcal{D}\cup\ov \D} \widehat{\mathcal{K}}(z,z)d^2z=0$.  We conclude that the  Hilbert-Carleman determinant \eqref{eq:tau-inf} coincides with the standard Fredholm determinant. More pragmatically, the series defining the two determinants coincide because $\mathcal K$ is identically zero on the diagonal.
\end{remark}

\noindent
{\it Proof of Theorem~\ref{th:FD}}.
Consider the operator $\mathcal{M}_{N}$ defined in \eqref{defMMd}.
It is convenient to introduce the two new operators $\mathcal{B}_{N}:=(\ov{\mathcal{M}}_{N}^{\dag}\circ \mathcal{M}_N)$ and $\ov{\mathcal{B}}_{N}:=( \mathcal{M}_{N}^{\dag} \circ \ov{\mathcal{M}}_{N} )$ on $L^2([x,+\infty))$ to itself which then  have the explicit  form
\bea
  \label{eq:op_N}
 \mathcal{B}_{N}[u](y) =
  \int^{+\infty}_{x}\hat{\mathcal{B}}_{N}(y + s)u(s)\d s, \qquad \hat{\mathcal{B}}_{N}(s)= \sum_{k=1}^{N}c_{k}{\rm e}^{i z_k s + 2iz_k^2 t}.
\eea
Then the tau function  \eqref{eq:tau-N} reads 
\be
\tau_N= \text{det}\bigg[\1_{\C^{N}} + \Phi_N \ov{\Phi_N} \bigg] = \text{det}\bigg[\1_{\C^{N}} + \mathcal M_{_{N}} \mathcal M_{_N}^\dagger \ov{\mathcal M}_{_{N}} \ov{\mathcal M_{_N}^\dagger} \bigg] 
=\nn\\
= \text{det}\bigg[\id_{L^{2}([x,+\infty))} +\ov{\mathcal M_{_N}^\dagger}  \mathcal M_{_{N}} \mathcal M_{_N}^\dagger \ov{\mathcal M}_{_{N}} \bigg] =\text{det}\bigg[\id_{L^2([x,+\infty))} + \mathcal{B}_{N} \ov{\mathcal{B}_{N}}
\bigg]\,.
\ee
We now rescale the constants in the same way  as in \eqref{eq:norm-const}  and send $N \to +\infty$. 
  For $(x, t)$ in a compact set,  the kernels $\hat{\mathcal{B}}_{N}$ and $\hat{\ov{\mathcal{B}}}_N$, defined in~\eqref{eq:op_N}, converge uniformly as  $N \to +\infty$ to the kernels $\hat{\mathcal{B}}$ and $\hat{\ov{\mathcal{B}}}$ defined as
  \begin{equation}
    \label{eq:op_inf}
   \hat{\mathcal{B}}(s):=
    \iint_{\mathcal{D}}
    \beta(w)^2e^{i(ws + 2w^2t)}\frac{\d^2 w}{\pi}, \qquad
     \hat{\ov{\mathcal{B}}}(s):= \iint_{\mathcal{D}}\ov{\beta(w)}^2{\rm e}^{-i(\ov{w}s + 2\ov{w}^2t)}\frac{\d^2 w}{\pi}.
  \end{equation}
 The corresponding convolution operators $\mathcal B:L^2([x,+\infty)) \to L^2([x,+\infty))$ and $\ov{\mathcal B}:L^2([x,+\infty)) \to L^2([x,+\infty))$
 \be
 \label{Hankel}
 \mathcal B[u](y) := \int_{x}^\infty \widehat{ \mathcal B}(y+s) u(s)\d s,  \ \ \ 
  \ov{\mathcal B[u]}(y) := \int_{x}^\infty \widehat{ \ov{\mathcal B}}(y+s) u(s)\d s
 \ee
  are, in fact, one the adjoint of the other, $\ov{\mathcal B} = \mathcal B^\dagger$, (being that the convolution kernel is a Hankel operator). 
  Then, from standard results on the convergence of operators in trace-class norm and continuity of the Fredholm determinant~\cite{Simon}, the $\tau$-functions $\tau_{N}(x,t)$ converges,  for $x, t$ in a compact set and  as $N \to +\infty$,  to  the Fredholm determinant  $\tau(x,t)$ 
\begin{equation}
  \label{eq:new_tau}
  \tau(x,t):= {\rm det}(\id_{L^2([x,+\infty))} + \mathcal{B}\circ \ov{\mathcal{B}})= {\rm det}(\id_{L^2([x,+\infty))} + \mathcal{B}\circ {\mathcal{B}^\dagger}).
\end{equation}
We observe that $\mathcal{B}\circ {\mathcal{B}^\dagger}$ is a trace-class operator because it is the product of two Hilbert-Schmidt operators and it is also a positive operator.
This also shows that $$\tau(x,t)>0\qquad \forall x, t\in\R .$$  
  
  In order to connect this  Fredholm  determinant with the  Hilbert-Carleman determinant \eqref{tau} for the operator $\mathcal K$  defined in \eqref{eq:Kern-NLS},  we  decompose the operator $\mathcal{B}$  in two Hilbert-Schmidt operators 
\begin{equation*}
  \mathcal{B}[u]= (\mathcal{L} \circ \mathcal{F})[u]
\end{equation*}
 where $\mathcal{L}: L^{2}(\mathcal{D})\to L^{2}([x,+\infty)) $ and $\mathcal{F}: L^{2}([x,+\infty)) \to L^{2}(\mathcal{D}) $  are defined in the following way:
\begin{align}
 & \mathcal{L}[\varphi](s):=\iint_{\mathcal{D}}{\beta(w)}{\rm e}^{i(ws + w^2t)}\varphi(w)
 \frac{\d^2w}{\pi},\ \ \ \   
 \varphi(w) \in L^{2}(\mathcal{D}), \,\  s\in \left[ \right. x, +\infty \left. \right)\\
  & \mathcal{F}[f](w):={\beta(w)} \int^{+\infty}_{x}{\rm e}^{i(ws + w^2t)}f(s) \d s, \ \ \ \ \ \ w \in \mathcal{D}, \  f(s) \in L^{2}( \left[ \right. x, +\infty \left. \right)  ).
\end{align}
The same is valid also for the operator ${\ov{\mathcal{B}}} = {\mathcal B}^\dagger$.  Now using cyclicity of the determinant,
we can rewrite \eqref{eq:new_tau}  in the form 
\begin{align*}
  &\tau(x,t)= {\rm det}(\id_{L^2([x,+\infty))} + \mathcal{B}\circ \ov{\mathcal{B}})
  = {\rm det}(\id_{L^2(\mathcal{D})}+\ov{\mathcal{P}}\circ{\mathcal{P}}),
\end{align*}
where $\mathcal{P} =  \mathcal{L}^{\dag} \circ \mathcal L: L^{2}(\mathcal{D}) \to L^{2}(\mathcal{D})$ and $\ov{\mathcal{P}}=  {\mathcal F} \circ \mathcal{F}^{\dag} : L^{2}(\mathcal{D}) \to L^{2}(\mathcal{D})$ are given, after a short computation,  by (recall that $\theta(z;x,t) = ixz + it z^2$)
\begin{align}
\label{323}
  \mathcal{P}[\varphi](z) &= i\iint_{\mathcal{D}}\frac{{\ov{\beta(z)}\beta(w)}
  {\rm e}^{- \theta(\ov z;x,t)+ \theta(w;x,t)}
  }{(w - \ov{z})}\varphi(w)\frac{\d^2 w}{\pi}\\
  \ov{\mathcal{P}}[\varphi](z)&= -i\iint_{\mathcal{D}}\frac{{\beta(z)\ov{\beta(w)}}
   { \rm e}^{  \theta(z;x,t) - \theta(\ov w;x,t)}
   }{(\ov{w}- z)}\varphi(w)\frac{\d^2 w}{\pi}.
\end{align}
Both the operators  $ \mathcal{P}$ and $ \ov{\mathcal{P}}$ are trace class, being the product of two Hilbert-Schmidt operators.
The spaces $L^2(\D)$ and $L^2(\ov \D)$ are clearly isometric and we can interpret $\mathcal P, \ov {\mathcal P}$ as maps $\mathcal P:L^2(\D)\to L^2(\ov \D)$ and $\ov {\mathcal P}: L^2(\ov \D)\to L^2(\D)$ (we use the same symbols) given by the similar formulas 
\begin{align}
  \mathcal{P}[\varphi](z) &= i\iint_{\mathcal{D}}\frac{{{\beta^{*}(z)}\beta(w)}
  {\rm e}^{-\theta(z;x,t)+ \theta(w;x,t)}
  }{(w - {z})}\varphi(w)\frac{\d^2 w}{\pi}, \ \ \ \ z\in \ov \D,\ \varphi\in L^2(\D)\\
  \ov{\mathcal{P}}[\psi](z)&= -i\iint_{\ov{\mathcal{D}}}\frac{{\beta(z){\beta^{*}(w)}}
   { \rm e}^{  \theta(z;x,t) - \theta( w;x,t)}
   }{({w}- z)}\psi(w)\frac{\d^2 w}{\pi},\ \ \ \ \ z\in \D,\ \psi\in L^2(\ov \D).
\end{align}
We thus have
\begin{equation}
\label{tau_3.25}
\begin{split}
\tau(x,t)&= \text{\rm det}(\id_{L^2(\mathcal{D})} +\ov{\mathcal{P}} \circ {\mathcal{P}}) ={\rm det}
  \begin{bmatrix}
    \id_{L^2(\mathcal{D})}+\ov{\mathcal{P}} \circ \mathcal{P} & i \ov{\mathcal{P}}\\
    0  & \id_{L^2(\ov \D)}
  \end{bmatrix}\\
&={\rm det}\left(
  \begin{bmatrix}
    \id_{L^2(\mathcal{D})} & i \ov{\mathcal{P}}\\
    i \mathcal{P}   & \id_{L^2(\ov \D)}
  \end{bmatrix}\begin{bmatrix}
    \id_{L^2(\mathcal{D})} &0\\
   - i \mathcal{P}   & \id_{L^2(\ov \D)}
  \end{bmatrix}\right)={\rm det}
  \begin{bmatrix}
    \id_{L^2(\mathcal{D})} & i \ov{\mathcal{P}}\\
    i \mathcal{P}   & \id_{L^2(\ov \D)}
  \end{bmatrix}\,.
  \end{split}
  \end{equation}
With this new understanding, we  define the trace class operator $\mathcal{K}$ on $L^2(\D\cup\ov \D) \simeq L^2(\D)\oplus L^2(\ov \D)$ as:
\begin{equation}
  \label{eq:def-K}
  \mathcal{K}:= -i \mathcal{P} -i \ov{\mathcal{P}}
\end{equation}
that has kernel $\widehat{\mathcal{K}}(z,w)$ as in \eqref{eq:Kth}. Note that $\widehat{\mathcal{K}}(z,z)=0$ and therefore 
$\mbox{Tr} \mathcal{K}=0$.
Then  the  Hilbert-Carleman  determinant   of $ \mathcal{K}$ is  a Fredholm determinant  and  now we show that it coincides with $\tau(x,t)$  in \eqref{tau_3.25} as follows:
\begin{align}
 & \text{\rm det}_2(\id_{L^2(\mathcal{D} \cup \ov{\mathcal{D}})} - \mathcal{K})=\det(\id_{L^2(\mathcal{D} \cup \ov{\mathcal{D}})} - \mathcal{K})\\
 & = {\rm det}
  \begin{pmatrix}
    \id_{L^2(\mathcal{D})} & i \ov{\mathcal{P}}\\
    i \mathcal{P}   & \id_{L^2(\ov \D)}
  \end{pmatrix}=  \det(\id_{L^2(\mathcal{D})} +\ov{\mathcal{P}} \circ {\mathcal{P}}) = \tau(x,t)>0.
\end{align}
The non vanishing of the $\tau$-function  
guarantees the existence of a  solution of the $\ov\partial$-problem \eqref{eq:d_bar}     $\forall x, t\in\R$ by Theorem~\ref{Theorem3.1}.\hfill$\square$


The $\tau$-function $\tau(x,t)$ depends smoothly on $x$ because the kernel is a smooth function of $x$ and the operator $\mathcal{K}$  is acting on a compact domain $\mathcal{D}\cup\overline{\mathcal{D}}$.
We conclude from the above theorem, that we can write the solution of the NLS equation originating from the   solution of the $\ov\partial$-problem \eqref{eq:d_bar}  using the $\tau$-function  of the problem  namely
\begin{equation}
\label{psi_fredholm}
 |\psi_{N}(x,t)|^2 = 2i\partial^2_x \log\det(\id_{L^2(\mathcal{D} \cup \ov{\mathcal{D}})} - \mathcal{K}).
 \end{equation}
\begin{remark}
The class of  solutions obtained  from \eqref{psi_fredholm}  is in general  different from the class studied   in   \cite{P} and represented via a Fredholm determinant of a kernel acting on contours.
From \eqref{Hankel} and \eqref{eq:new_tau} the kernel of the $\tau$-function is the  composition of Hankel operators acting on domains of the complex plane.  Recently Bothner \cite{BB} and A. Krajenbrink \cite{KlD}  enlarged,  in a different direction with respect to our case,  the class of Hankel composition operators to obtain new class of solutions of the modified Korteweg di Vries equation.  Applications are obtained in \cite{BCT}, \cite{CCR}.
 \end{remark}

\section{Step-like oscillatory initial data}
\label{sec4}
{\color{black}For certain class of domains $\D$  called generalized quadrature domains  and $\beta$ analytic,  the  $\ov\partial$-problem can be reduced to a Riemann problem. A generalized quadrature domain ${\mathcal D}$  is   simply connected and the  boundary of  ${\mathcal D}$ is sufficiently smooth so that it can be  described by the so--called   Schwarz function  $S(z)$  of the domain ${\mathcal D}$  through the equation
\[
\overline{z}=S(z).
\]
The condition to be a generalized quadraure domain is that the Schwarz function of the domain  \cite{Gustafsson} admits  an  analytic extension to  an open and dense maximal sub-domain $\D^0\subset \D$ and  $\mathcal{L}:=\D\setminus \D^0$  consists of a {\it mother-body}, i.e.,  a collection of smooth arcs.  Using Stoke theorem and the Cauchy theorem, the $\ov\partial$-problem can be reduced to a Riemann problem with discontinuities along $\mathcal{L}$.
We consider the case when $\D$  is an  ellipse. In this case we can reduce the $\ov \pa$-problem to a Riemann-Hilbert Problem (RHP) on  two segments, one   connecting the foci of the ellipse and  the other  its Schwartz reflection as we now explain. 
The degeneration of this case to the circle was already considered in \cite{BGO}.}
For the sake of simplicity we  assume  that the focal points of the ellipse $E_1$ and $E_2$  are situated on  the imaginary axis, i.e. $E_1= i\alpha_1$ and $E_2=i\alpha_2$ with  $\alpha_2>\alpha_1>0$. The equation of the ellipse is
\[
\sqrt{\re(z)^2+(\im(z)-\alpha_1)^2}+\sqrt{\re(z)^2+(\im(z)-\alpha_2)^2}=2\rho>0,
\] where $\rho$ is chosen sufficiently small so that $\D$ lies in the upper half plane (see Figure~\ref{fig:dom}).
\begin{figure}[h!]
  \centering
     \begin{tikzpicture}
    \draw[very thick, ->] (0,-3)--(0,3) node[left] {$\im(z)$};
    \draw[very thick, ->] (-2,0)--(2,0) node[below right] {$\re(z)$};
    \draw[fill=blue!30, ultra thin] (0,1.5) ellipse (0.8 and 0.9);
    \draw (1.5,1.5) node{$\mathcal{D}$};
    \filldraw[black] (0,2) circle (1pt) node[below left] {\small$i\alpha_2$};
    \filldraw[black] (0,1.) circle (1pt) node[above right] {\small$i\alpha_1$};
    \draw[fill=blue!30, ultra thin] (0,-1.5) ellipse (0.8 and 0.9);
    \draw (1.5,-1.5) node{$\ov{\mathcal{D}}$};
    \filldraw[black] (0,-2) circle (1pt) node[above left] {\small -$i\alpha_2$};
    \filldraw[black] (0,-1.) circle (1pt) node[below right] {\small -$i\alpha_1$};
    \draw[thick,red] (0,1.)--(0,2.);
    \draw[thick,red] (0,-2.)--(0,-1.);
  \end{tikzpicture}
  \caption{ The domains $\D$ and $\overline{\D}$.}\label{fig:dom}
\end{figure}
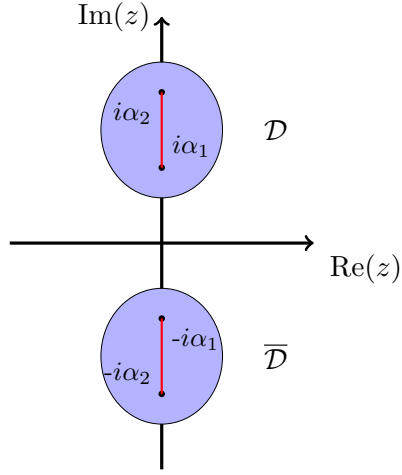

 %
%


   We can apply Green's theorem for $z\notin \D$ in the formulas below, 
 and obtain
 \begin{equation}
 \label{eq:int_prob1}
  \iint\limits_{\D}\frac{{\rm e}^{2\theta(w,x,t)} \beta(w)^2\d^2w}{\pi(z-w)}=
   \int_{\partial\D}\frac{\beta(w)^2\overline{w}{\rm e}^{2\theta(w;x,t)}}{ z-w} \frac{\d w}{2\pi  i }
   \end{equation}
   and similar expressions for the integral over $\overline{\D}$.
  We consider the Schwartz function $S$ of the ellipse
\begin{equation}
\begin{split}
  \ov{z} = S(z)&=\left(1 -2\frac{\rho^2}{c^2}\right)(z-i y_0)+2\frac{\rho}{c^2}\sqrt{\rho^2-c^2}\tilde{S}(z)
\end{split}
\end{equation}
where $\tilde{S}(z):= \sqrt{ (z-i\alpha_1)(z-i\alpha_2)}$,  $ y_0:=\frac {\alpha_1+\alpha_2}2$ and   $c := \frac{\alpha_2-\alpha_1}2$. 
  The  Schwartz function $S$  is analytic  in $\C \setminus \mathcal{I}$, where  $\mathcal{I}:=[i\alpha_1,i\alpha_2]$, with boundary values $S_{\pm}(w)$. 
   For $z\notin \D$ (or $z \notin \ov{\D}$), 
the integral along  the boundary $\partial \D$  ($\partial \overline{\D}$) of the ellipse in \eqref{eq:int_prob1}  can be deformed to a line integral  on  the segment $\mathcal{I}=[i\alpha_1,i\alpha_2]$  ($\overline{\mathcal{I}}:=[-i\alpha_2,-i\alpha_1]$), namely
\begin{align*}
 \int_{\partial \D}\frac{\beta(w)^2\overline{w}{\rm e}^{2\theta(w;x,t)}}{ z-w} \frac{\d w}{2\pi  i }&=
  \int_{\partial\D}\frac{\beta(w)^2S(w){\rm e}^{2\theta(w;x,t)}}{ z-w} \frac{\d w}{2\pi  i }\\
&= 
 \int_{\mathcal{I}}\frac{\beta(w)^2\delta S(w){\rm e}^{2\theta(w;x,t)}}{ z-w} \frac{\d w}{2\pi  i }
\end{align*}
where $\delta S(z)= S_{+}(z) - S_{-}(z)$. 
We define the matrix $T(z)$ as 
\begin{equation}
T(z):=\left\{
\begin{array}{ll}
&Y^\infty(z),\quad z\in\mathbb{C}\backslash\{D_{\gamma_+}\cup D_{\gamma_-}\}\\
&Y^\infty(z)J(z),\quad z\in D_{\gamma_+}\cup D_{\gamma_-}
\end{array}
\right.
\end{equation}
with $Y^\infty$ as in \eqref{eq:new_R-H_2}, $D_{\gamma_\pm}$ the interior regions of the contours $\gamma_\pm$, respectively,  and    
\begin{equation*}
  \begin{split}
J(z)= \begin{pmatrix}
\hspace{-20pt}    1 &\hspace{-60pt}\displaystyle \int\limits_{\overline{\mathcal{I}}}\frac{(\beta^*(w))^2\delta S^*(w){\rm e}^{-2\theta(w;x,t)}}{ w-z} \frac{\d w}{2\pi  i }{\bf 1}_{{D}_{\gamma_-}}(z) \\
\displaystyle   \int\limits_{\mathcal{I}}\frac{\beta(w)^2\delta S(w){\rm e}^{2\theta(w;x,t)}}{ z-w} \frac{\d w}{2\pi  i }{\bf 1}_{{D}_{\gamma_+}}(z)& 1
  \end{pmatrix}.
  \end{split}
\end{equation*}
This new  matrix $T(z)$ extends analytically across  $\gamma_+\cup\gamma_-$. The matrix $J(z)$ has a discontinuity on  $\I\cup\overline{\I}$ so that   $T(z)$ solves the RHP described below:
\begin{problem} \label{RHP41} The matrix $T(z)$ is analytic and analytically invertible  in  $\mathbb{C}\backslash \{\I\cup \overline{\I}\}$  with jump conditions
\begin{equation}
\begin{split}
  \label{eq:ell_R-H-prob}
  T_{+}(z)&=T_{-}(z)J^{-1}_-J_+=T_{-}(z){\rm e}^{\theta(z;x,t) \sigma_{3}}V_{0}(z){\rm e}^{-\theta(z;x,t) \sigma_{3}},\quad z\in
  \I\cup \overline{\I}\\
    V_{0}(z)&= \begin{pmatrix}
    1 &{\bf 1}_{\I^*} \delta S^*(z) (\beta^*(z))^2 \\
    -{\bf 1}_{\I}\delta S(z) \beta(z)^2 & 1
  \end{pmatrix},
\end{split}
\end{equation}
 and
$T(z)={\bf 1}+O(z^{-1})$, as $z\to\infty$.
\end{problem}
We can find the same Problem~\eqref{RHP41} also while studying the infinite soliton limit when the  spectral points are  distributed along the segments  $\I\cup\overline{\I}$.
The solution of the NLS equation is recovered from the solution of the RHP for $T(z;x,t)$ by the relation
\[
\psi_N(x,t)=2i \lim_{z \to \infty}z(T(z;x,t))_{12}.
\]

\subsection{Charachterization of the initial datum}

We now focus on the Problem~\eqref{RHP41} at $t=0$. Let us define the  function ${r} (z):=\delta S(z)\beta(z)^2$,
so that  the  RHP problem becomes
\begin{equation}
\begin{split}
  \label{eq:ell_z}
  & T_{+}(z)=T_{-}(z)
  \begin{pmatrix}
    1 & {r} ^*(z)e^{-2 izx}{\bf 1}_{\overline{\mathcal{I}}}(z)\\
    -{r} (z)e^{2 izx}{\bf 1}_{\mathcal{I}}(z) & 1
  \end{pmatrix}\qquad z \in \mathcal{I}\cup \overline{\mathcal{I}},\\
  &T(z)={\bf 1}+O(z^{-1}),  \mbox{   as $z\to\infty$.}
  \end{split}
\end{equation}
We also introduce the  elliptic  curve 
    \begin{equation}
    \label{C}
      \mathcal{C}=\{(w,z) \in \C^2 | w^2= R(z)=(z^2+\alpha_1^2)(z^2 + \alpha_2^2)\}. 
    \end{equation}
    The projection $\pi:\mathcal{C}\to\C$,    $\pi((z,w))=z$, realizes $  \mathcal{C}$ as a two sheeted covering of the complex plane of the radical $\sqrt{(z^2+\alpha_1^2)(z^2 + \alpha_2^2)}$, with branch-cuts along the segments $\I \cup \overline{\I}$ 
and the determination chosen so that the radical behaves as $z^2$ near $\infty$. 
     To  make the curve $\mathcal{C}$ a complex Riemann surface $\mathcal{S}$ we  add the two points at infinity  $\infty^{1,2}$ where  $\infty^1$ is     on the first sheet ($+$ sign of the radical) and  $\infty^2$ is the second sheet ($-$ sign).  When  indicating a point on the surface as $z$, it is understood that   $(z,w)\in\mathcal{C}$ is a point on the first sheet.
    We consider the  normalized holomorphic one form
  
    \[
    \omega=\left(\oint_{\alpha}\frac{dz}{w}\right)^{-1}\frac{dz}{w}= \frac{i(\alpha_{2}+\alpha_1)}{4 K(m)}\frac{dz}{w},  \quad m=\frac{ 4\alpha_1\alpha_2}{(\alpha_1+\alpha_2)^2}
    \]
    where $K(m)=\int_{0}^1\dfrac{ds}{\sqrt{1-s^2}\sqrt{1-m s^2}}$ is the complete elliptic integral of the first kind.
     We define the modulus 
   \[
   \tau:=\oint_{\beta}\omega=\frac{i(\alpha_1+\alpha_{2})}{2 K(m)}  \int_{i\alpha_2}^{i\alpha_1}\frac{dz}{\sqrt{R(z)_+}}=\frac{iK(m')}{ K(m)},\quad m'=1-m,
   \]
   where the  $\alpha$ and  $\beta$ cycles are defined as in the Figure~\ref{fig_cycles}.
    \begin{figure}[t]
    \centering
    \includegraphics[width=0.5\textwidth]{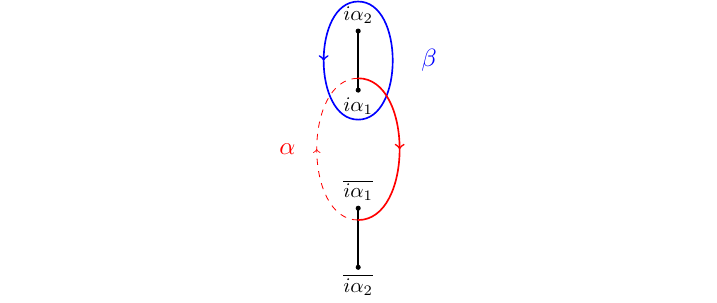}
    \caption{ 
    The homology  basis }\label{fig_cycles}
  \end{figure}
 We recall  the definition of the   Jacobi  theta  function  $\vartheta: \C \to \C$  with modulus $\tau$ 
\begin{equation}
  \label{eq:Jac-theta}
  \vartheta(z,\tau):= \sum_{n \in \mathbb{Z}} e^{i\pi n^{2}\tau + 2i\pi n z}
\end{equation}
The goal of the section is to prove the following Theorem:
\begin{theorem}
\label{thmPsi0}
   The solution of the Riemann-Hilbert Problem~\eqref{eq:ell_z} generates an initial datum $\psi_{0}(x)$  of the NLS equation~\eqref{NLS} that is step-like oscillatory with the following behaviours at $x \to \pm \infty$ : 
    \begin{equation}
     \label{eq:NLS-t}
     \psi_{0}(x)=\left\{\begin{matrix}
         \mathcal{O}(\e^{-c_+ x}) \text{ as } x \to +\infty\\
         -i\e^{2ig_{\infty}x+2i\phi_\infty}(\alpha_{2} - \alpha_{1})\frac{\vartheta(0;\tau)}{\vartheta(\frac{1}{2};\tau)}\frac{\vartheta(\frac{\Omega x  - i\Delta}{2 \pi} +\frac{1}{2};\tau)}{\vartheta(\frac{\Omega x - i\Delta}{2 \pi })} + \mathcal{O}(\e^{c_- x})  \text{ as } x \to -\infty
       \end{matrix}\right.
   \end{equation}
 where $c_{\pm}$  are positive constants, $g_{\infty},\phi_\infty\in \mathbb{R}$ are constants, $\vartheta(z;\tau)$ is the {\it Jacobi theta function} defined in \eqref{eq:Jac-theta} and  
   \begin{equation}
      \label{eq:Omega}
      \Omega =- \frac{\pi(\alpha_1+\alpha_2)}{ K(m)}\in\R,
    \end{equation}
    \begin{equation}
      \label{eq:delta}
      \Delta=-\frac{i(\alpha_{2}+\alpha_1)}{2 K(m)}\left[\int_{i\alpha_1}^{i\alpha_2}\frac{\log {r} (\zeta)}{\sqrt{R(\zeta)}_+}d \zeta - \int_{-i\alpha_2}^{-i\alpha_1}\frac{\log {r} ^{*}(\zeta)}{\sqrt{R(\zeta)}_+}d \zeta \right]\in i\R.
    \end{equation}
 where $\sqrt{R(z)}$  is a multivalued complex function, analytic in $\C \setminus\{\mathcal{I} \cup \overline{\mathcal{I}}\}$  and positive in the interval $(-i\alpha_1,i\alpha_1)$  and $\sqrt{R(z)}_\pm$
   denotes the value  on the positive/negative  side of the oriented intervals $\mathcal{I}$   and $ \overline{\mathcal{I}}$.
 \end{theorem}

In order to prove the above theorem we follow the established procedure of the Deift--Zhou steepest descent analysis \cite{DZ} which involves several reformulations of the problem:
\begin{itemize}
\item introduction of piece-wise analytic scalar functions  $g$ and $f$;
\item  first transformation: $T(z;x)\to T^{(1)}(z;x)$;
\item  second transformation:  $T^{(1)}(z;x)\to T^{(2)}(z;x)$;
\item construction of the outer parametrix $X(z;x)$ and estimate of the remainder term  $\mathcal{E}=T^{(2)}(z;x)X^{-1}(z;x)$.
\end{itemize}

\paragraph{\bf The $g$-function and the  $f$-function.}
We introduce the functions  $g$ and $ f$  defined by:
\bea
  \label{eq:g_funct}
  g(z)&:= -z +\int_{i\alpha_{2}}^{z}\frac{\zeta^{2} + \kappa}{\sqrt{R(\zeta)}} d \zeta
\\
      \label{eq:kappa}
      \kappa&:= -\frac{\int_{-i\alpha_1}^{i\alpha_1}\frac{\zeta^{2}\d \zeta}{\sqrt{R(\zeta)}}}{\int_{-i\alpha_1}^{i\alpha_1}\frac{\d \zeta}{\sqrt{R(\zeta)}} }= \alpha_{2}^{2}\left(1-\frac{E(m_1)}{K(m_1)}\right),\quad m_1=\frac{\eta_1^2}{\eta_2^2}\,;
\eea
where $E(m_1)= \int_{0}^{1}\sqrt{\frac{1 - m_1 s^2}{1 - s^2}}d s$ is the complete elliptic integral of the second kind.
The phase $\Omega$ in \eqref{eq:Omega} is simply
\[
\Omega:=\oint_{\beta}\frac{\zeta^{2} + \kappa}{\sqrt{R(\zeta)}} d \zeta=- \frac{\pi(\alpha_1+\alpha_2)}{ K(m)}\in\R,
\]
that can be easily be obtained from Riemann bilinear relations \cite{Mumford2007,Mumford2008}.
\begin{align}
  \label{eq:f_funct}
    f(z):= \exp\Bigl[\frac{R(z)}{2\pi i} \Bigl(&-\int_{i\alpha_1}^{i\alpha_{2}}\frac{\log{{r} (\zeta)}}{\sqrt{R(\zeta)}_{+}(\zeta - z)}d \zeta + \int_{-i\alpha_2}^{-i\alpha_1}\frac{\log{{r} ^{*}(\zeta)}}{\sqrt{R(\zeta)}_{+}(\zeta - z)}d \zeta + \nn\\
        & + \int_{-i\alpha_1}^{i\alpha_1}\frac{\Delta}{\sqrt{R(\zeta)}(\zeta - z)}d \zeta  \Bigr) \Bigr]  ;
\end{align}

  The following properties of the functions $f,g$ follow from the above definition.
  \begin{lemma}
\label{propertygf}
The functions $g$ and $f$ defined in \eqref{eq:g_funct} and  \eqref{eq:f_funct}
satisfy:
\begin{enumerate}
\item Schwartz symmetry:
 \begin{equation*}
    \overline{g(\bar{z})}=g(z) \quad \overline{f(\bar{z})}=f^{-1}(z);
  \end{equation*}
  \item as $z \to \infty$
  \begin{equation}
    \label{eq:cond-g}
    \begin{split}
&g(z) = g_\infty+ \left[\frac{\alpha_1^{2} +\alpha_2^{2}}{2} - \kappa \right]\frac{1}{z} + \mathcal{O}\left(z^{-2}\right) \\
 &f(z) = \e^{i\phi_\infty}+ \mathcal{O}\left(z^{-2}\right);
 \end{split}
\end{equation}
where $g_{\infty} \in \mathbb{R}$ and  $\phi_\infty \in \mathbb{R}$ are constants;
\item  for $z$ near the endpoints $\pm i\alpha_j$, for $j=1,2$, 
  \begin{align*}
    g(z) + z =O((z - i\alpha_2)^{1/2}) &\quad g(z) + z= O((z + i\alpha_2)^{1/2});\\
    g_+(z) + z = \frac{\Omega}{2} + O((z - i\alpha_1)^{1/2}) &\quad g_+(z) + z  =\frac{\Omega}{2} + O( (z + i\alpha_1)^{1/2});
  \end{align*}
  
\item the function $g(z)$ solves the following scalar Riemann-Hilbert Problem:
\begin{equation}
  \label{eq:g_R-H}
  \begin{split}
    & g_{+}(z) + g_{-}(z) = -2z \text{ for } z \in \mathcal{I} \cup \overline{\mathcal{I}}\\
    & g_{+}(z) - g_{-}(z) = \Omega \text{ for } z \in [-i\alpha_{1};i\alpha_{1}]\\
  \end{split}
\end{equation}
  with the real constant $\Omega$ as in \eqref{eq:Omega}.
\item  the function $f(z)$ satisfies the scalar Riemann-Hilbert problem:
\begin{equation}
  \label{eq:f_R-H}
  \begin{split}
    & f_{-}(z)f_{+}(z)={r} ^{-1}(z) \text{ for } z \in \mathcal{I}\\
    & f_{-}(z)f_{+}(z)= {r} ^{*}(z) \text{ for } z \in \overline{\mathcal{I}}\\
    & \frac{f_{+}(z)}{f_{-}(z)}= e^{\Delta} \text{ for } z \in [-i\alpha_{1};i\alpha_{1}]\\
  \end{split}
\end{equation}
with $\Delta$ defined in \eqref{eq:delta}.
\end{enumerate}
\end{lemma}
%
\noindent
\paragraph{First transformation: $\boldsymbol{T(z;x)\to T^{(1)}(z;x)}$.}
With the functions $f$ and $g$ we define the matrix
\begin{equation}
  \label{eq:nsd_trans}
  T^{(1)}(z;x)= \e^{-i(g_{\infty}x+\phi_\infty)\sigma_3}T(z;x)e^{ig(z)x \sigma_3}f(z)^{\sigma_3}.
\end{equation}
As a consequence of the transformation \eqref{eq:nsd_trans}  and of the properties of $f, g$ established in the Lemma \ref{propertygf} we obtain that 
\begin{equation}
\label{psi0T1}
\psi_0(x)=2i\e^{i(g_{\infty}x+\phi_\infty)}\lim_{z\to\infty} \left[zT^{(1)}_{12}(z;x)e^{ig(z)x}f(z)\right],
\end{equation}
and the matrix $T^{(1)}(z;x)$ satisfies  the following RHP.
\begin{problem}
\label{RHPGamma1}
To find a $2\times 2 $ matrix $T^{(1)}(z;x,t)$  analytic and invertible  in $\C\setminus [-i\alpha_2,i\alpha_2]$  with the following boundary value conditions:
\begin{equation}
  \label{eq:new_ell_z}
  T^{(1)}_{+}(z)=T^{(1)}_{-}(z)V^{(1)}(z)
\end{equation}
where
  \begin{equation}
    \label{eq:jump_m_new}
    V^{(1)}(z) =\left\{
    \begin{split}
   & \begin{pmatrix}
      e^{ix(g_+ - g_-)}\frac{f_{+}}{f_{-}} & {\bf 1}_{\overline{\mathcal{I}}}\\
      -{\bf 1}_{\mathcal{I}} &   e^{-ix(g_+ - g_-)}\frac{f_{-}}{f_{+}}
     \end{pmatrix} \text{ for } z \in \mathcal{I} \cup \overline{\mathcal{I}} \\
   &
   \begin{pmatrix}
     e^{ix\Omega + \Delta} & 0\\
     0 & e^{-ix\Omega -\Delta}
   \end{pmatrix} \text{ for } z \in [-i\alpha_1;i\alpha_1]\\
    \end{split} \right.
  \end{equation}
  and
  
 $$T^{(1)}(z) = \1+O(z^{-1}),\quad \mbox{as $z\to\infty$}.$$
%
 \end{problem}
 %

\paragraph{\bf The Opening of the lenses.} We factorize the jump matrices~\eqref{eq:jump_m_new} for $z \in \mathcal{I} \cup \overline{\mathcal{I}}$ as follows:
  \begin{itemize}
 \item for $ z \in \overline{\mathcal{I}}$
  \begin{equation*}
      V^{(1)}(z) =
       \begin{pmatrix}
          1 & 0\\
        e^{2ix(g_- + z)}\frac{(f_{-}(z))^2}{{r} (z)} & 1 
      \end{pmatrix}
       \begin{pmatrix}
        0 & 1\\
        -1 & 0
      \end{pmatrix}
      \begin{pmatrix}
        1 & 0\\
         e^{2ix(g_+ +z)}\frac{(f_{+}(z))^2}{{r} (z)} & 1
       \end{pmatrix}
  \end{equation*}
 \item for $ z \in \mathcal{I}$ 
  \begin{equation*}
     V^{(1)}(z)=
      \begin{pmatrix}
        1 &-\frac{ e^{-2ix(g_- +z)}}{{r} ^{*}(z)(f_{-}(z))^2}\\
        0 & 1
      \end{pmatrix}
      \begin{pmatrix}
        0 & 1\\
        -1 & 0
      \end{pmatrix}
      \begin{pmatrix}
        1 & -\frac{ e^{-2ix(g_+ +z )}}{{r} ^{*}(z)(f_{+}(z))^2}\\
        0 & 1
      \end{pmatrix}.
  \end{equation*}
\end{itemize}
We define with $\mathcal{U}_{\pm}$   the open set on the left ($+$) and the right  ($-$) of the segment $\mathcal{I}$ as  shown  in Figure~\ref{fig:2}. In the same way, we define $\overline{\mathcal{U}_{\pm}}$ as the complex conjugate of  
$\mathcal{U}_{\pm}$   respectively.

  The function $\delta S(z) = S_{+}(z) - S_{-}(z)$ is in principle defined only on $\mathcal I$;  however one can  define an  analytic extension  to  the left and in the right of $\mathcal I$.
  The function $\beta^2(z)$ is assumed analytic in the neighbourhood of $\mathcal I$ and thus we can analytically extend ${r}(z)$ to the regions $\mathcal U_\pm$. We denote by  $\hat{r}(z)$ the analytic extension of ${r}(z)$. This procedure can be repeated by Schwartz symmetry also to $r^*(z)$.

\noindent
\paragraph{Second transformation: $\boldsymbol{ T^{(1)}(z,x)\to  T^{(2)}(z,x)}$.}
We  introduce a further  transformation of the problem:
\begin{equation}
  \label{eq:S_transf}
  T^{(2)}(z,x)= T^{(1)}(z,x)G(z,x),
\end{equation}
where
  \begin{equation}
    G(z,x)=\left\{
      \begin{split}
       & \begin{pmatrix}
          1 & \frac{e^{-2ix(g(z) +z)}}{\hat{r}(z)f^{2}(z)}\\
         0 & 1   
       \end{pmatrix} \text{ for } z \in \mathcal{U}_{+}; \\
       &\begin{pmatrix}
           1 &  -\frac{e^{-2ix(g(z)+z)}}{\hat{r}(z)f^{2}(z)}\\
         0 & 1
      \end{pmatrix} \text{ for } z \in \mathcal{U}_{-};\\
      &
      \begin{pmatrix}
         1 & 0\\
        - e^{2ix(g(z)+z)}\frac{f^{2}(z)}{\hat{r} ^*(z)} & 1
      \end{pmatrix} \text{ for } z \in \overline{\mathcal{U}_{+}}; \\
      &
      \begin{pmatrix}
        1 & 0\\
         e^{2ix(g(z) + z)}\frac{f^{2}(z)}{\hat{r} ^*(z)} & 1
       \end{pmatrix} \text{ for } z \in\overline{\mathcal{U}_{-}}; \\
      &\;\; \1 \qquad \text{ otherwise. }
      \end{split}\right.
  \end{equation}
  \begin{figure}[h!]
    \centering
    \includegraphics[width=0.7\textwidth]{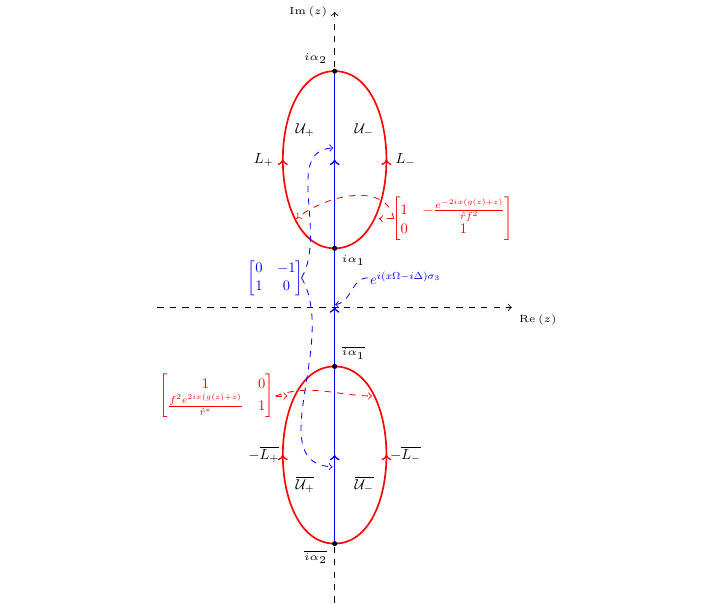}
    \caption{ The lenses $\mathcal{U}_{+}$ and $\mathcal{U}_{-}$  }\label{fig:2}
  \end{figure}
  As a consequence of the  above  transformation  we obtain  from \eqref{psi0T1} that 
\begin{equation}
\label{psi0T2}
\psi_0(x)=2i\e^{i(g_{\infty}x+\phi_\infty)}\lim_{z\to\infty} \left[zT^{(2)}_{12}(z;x)e^{ig(z)x}f(z)\right],
\end{equation}
and  the  matrix  $T^{(2)}$ satisfies  a new Riemann-Hilbert Problem with discontinuities  not only on  $\mathcal{I} \cup \overline{\mathcal{I}}\cup[-i\alpha_1,i\alpha_1]$ but also on the boundaries of the lenses $\mathcal{U}_{\pm}$ and $\overline{\mathcal{U}_{\pm}}$, which we denote with $L_{\pm}$ and $\overline{L_{\pm}}$, namely 
    \begin{equation}
      \label{eq:fin-rhp}
      \begin{split}
      &T^{(2)}_{+}(z;x)=T^{(2)}_{-}(z;x)V^{(2)}(z;x),\\
      &T^{(2)}(z;x) = \1+O(z^{-1}),\quad \mbox{as $z\to\infty$},
      \end{split}
    \end{equation}
  where the matrix  $V^{(2)}$   takes the following form: 
    \begin{equation*}
      V^{(2)}(z;x)=\left\{
        \begin{array}{ll}
          \begin{pmatrix}
            0 & 1\\
            -1 & 0
          \end{pmatrix} &\text{ for } z \in \mathcal I\cup\mathcal I^*\\[15pt]
          \begin{pmatrix}
             e^{ix\Omega + \Delta} & 0\\
     0 & e^{-ix\Omega -\Delta}
          \end{pmatrix}& \text{ for } z \in [-i\alpha_1;i\alpha_1]\\
        \end{array}\right.
      \end{equation*}
      while on  the  contours $L_\pm$ and $\overline{L_\pm}$ 
 \begin{equation*}
       V^{(2)}(z;x) =\left\{
      \begin{array}{ll}
        \begin{pmatrix}
          1 & -\frac{ e^{-2ix(g(z) +z)}}{\hat{r}(z)f^{2}(z)}\\
       0 & 1   
     \end{pmatrix} &\text{ for } z \in L_{+}\cup L_{-}, \\[15pt]
      \begin{pmatrix}
        1 & 0\\
         e^{2ix(g(z) + z)}\frac{f^{2}(z)}{\hat{r}^{*}(z)} & 1
       \end{pmatrix} &\text{ for }z \in \overline{L_{+}}\cup  \overline{L_{-}};\\
      \end{array}\right.
  \end{equation*}
 so that the $z$--dependent oscillatory part appears   only along the outer boundary of the lenses.  

 Now we need to study the sign of $\im(g(z) + z)$ around the lenses.

 \begin{lemma}
   The function $g(z)$ satisfy the following inequalities:
   \begin{equation}
     \label{eq:ineq}
     \begin{split}
       &\im(g(z) + z) > 0 \text{ for } z \in\C_+\setminus \{i\alpha_{1};i\alpha_{2}\} \\
       &\im(g(z) + z) < 0 \text{ for } z \in \C_-\setminus \{-i\alpha_2;-i\alpha_1\}
     \end{split}
   \end{equation}
 \end{lemma}
 \begin{proof}
The expression we are studying is (see \eqref{eq:g_funct})
\be
\Phi(z):= \im \int_{i\alpha_{2}}^{z}\frac{\zeta^{2} + \kappa}{\sqrt{R(\zeta)}} d \zeta
\ee
with $\kappa$ defined in \eqref{eq:kappa}.
One verifies that, while the integral is not single--valued on $\C \setminus \mathcal I \cup \overline{\I}$, its periods are purely real thanks to the definition of $\kappa$. Therefore $\Phi$ is a harmonic function on $\C \setminus \mathcal I \cup \overline{\I}$.  Next, one observes that $\Phi(z)$ is zero on $\mathcal I$ because both boundary values of the integral are real and in particular $\Phi$ is continuous across $\mathcal I\cup \overline{\I}$ (but not differentiable). Moreover for $z\in \R$ we also have $\Phi(z)\equiv 0$ because the integral is purely real. 
Finally, since the integrand is $1 + \mathcal O(\zeta^{-2}))$, the integral behaves as $z+\mathcal O(1)$ as $z\to\infty$ and hence $\Phi(z)$ has the same sign$\im z$ for large $z$.
By the extremum principle of harmonic functions we deduce that $\Phi$ is strictly positive in $\C_+\setminus \mathcal I$ and negative in $\C_-\setminus \overline{\I}$.
\end{proof}

Let $U_j$ be a neighbourhood of $i\alpha_j$,  and similarly $\overline{U_j}$ be a neighbourhood of  $-i\alpha_j$  with $j=1,2$. Let
 $ \Gamma=L_{+}\cup L_-\cup \overline{L_+}\cup   \overline{L_-}$.
  Lemma~\ref{eq:ineq}, implies that for $x \to -\infty$ the jump matrix $V^{(2)}(z;x)$   converges to the identity exponentially fast 
   for $z \in \hat{\Gamma}$ where $ \hat{\Gamma}=\Gamma\backslash\{\Gamma\cap\{U_1\cup U_2\cup \overline{U_1}\cup \overline{U_2}\}\}$.
   We arrive to the model problem for a matrix $X(z;x)$.
\begin{problem}
  \label{prob:mod-prob}
 Find a   matrix  $X(z): \C \to GL(2,\C)$, analytic and analytical invertible in $z\in \C \setminus [-i\alpha_2,i\alpha_2]$, with jump conditions 
 \begin{equation}
   \label{eq:inft_R-H}
   \begin{split}
     & X_{+}(z;x)= X_{-}(z;x)V_{X}(z;x)\\
     &V_{X}(z;x)=\left\{
        \begin{matrix}
          \begin{pmatrix}
            0 & 1\\
            -1 & 0
          \end{pmatrix} &\text{ for } z \in \I \cup \overline{\I}\\
          \begin{pmatrix}
            e^{ix\Omega + \Delta} & 0\\
     0 & e^{-ix\Omega -\Delta}
          \end{pmatrix} &\text{ for } z \in [-i\alpha_1;i\alpha_1]
        \end{matrix}\right.
   \end{split}
 \end{equation}
 with boundary condition at infinity
 \begin{equation*}
   X(z;x)= \1 + \mathcal{O}(z^{-1}), \text{ as } z \to \infty.
 \end{equation*}
\end{problem}
    The RHP~\ref{prob:mod-prob} is solved using   the Jacobi theta function \eqref{eq:Jac-theta} (see e.g.  \cite{DKZ}, \cite{Koroktin}, \cite{DKMVZ}, \cite{BT2017}).
We define the {\it Abel map} $u: \mathcal{C} \to \C$ as
    \begin{equation}
      \label{eq:abel}
  u(z,z_0)=\frac{i(\alpha_1+\alpha_{2})}{4 K(m)}\int^{(z,w)}_{(z_0,w_0)}\frac{d \zeta}{\sqrt{R(\zeta)}},
\end{equation}
where each point $(z,w)$ and $(z_0,w_0)$  belong to  the elliptic curve $\mathcal{C}$ in \eqref{C}.
By setting $(z_0,w_0)=\infty^1$, the Abel map~\eqref{eq:abel} has the following jump conditions along the segment $[-i\alpha_2,i\alpha_2]$
\begin{align}
  \label{eq:jump-u}
  &u_{+}(z;\infty^1) + u_{-}(z;\infty^1) = -\frac{1}{2} \quad  \text{ for } z \in \mathcal{I}; \nn\\
  &u_{+}(z;\infty^1) - u_{-}(z;\infty^1) = \tau \quad  \text{ for } z \in [-i\alpha_1;i\alpha_1];\\
   &u_{+}(z;\infty^1) + u_{-}(z;\infty^1) = \frac{1}{2} \quad  \text{ for } z \in \overline{\mathcal{I}}. \nn
\end{align}
Here   and below the point $(z,w)$ is simply referred to as $z$ and it is   understood to belong to the first sheet of the  surface $\mathcal{C}$.
Next we recall that the Jacobi theta function defined in \eqref{eq:Jac-theta} satisfies the periodicity conditions 
  \begin{equation}
  \label{period_theta}
\vartheta(z +j+ l\tau;\tau)= \vartheta(z;\tau)\e^{-2i\pi z l - i \pi l^2 \tau},\quad \text{for}  \;l,j \in \mathbb{Z}\,.
\end{equation}
 Using the  above property and the jump conditions  of $u(z,\infty^1)$~\eqref{eq:jump-u}, we can show that the solution of the model problem~\ref{prob:mod-prob} has the following form:
    \begin{equation}
      X(z;x)=
      \frac{\vartheta(0;\tau)}{2\vartheta(\epsilon;\tau)}\begin{pmatrix}
        \frac{\vartheta(u(z^{(1)},\infty^1) -\epsilon;\tau)}{\vartheta(u(z^{(1)},\infty^1);\tau)}\left(\phi(z) +\phi(z)^{-1}\right) &-i\frac{\vartheta(u(z^{(2)},\infty^1) -\epsilon;\tau)}{\vartheta(u(z^{(2)},\infty^1);\tau)}\left(\phi(z) -\phi(z)^{-1}\right)\\
        i\frac{\vartheta(u(z^{(1)},\infty^2) -\epsilon;\tau)}{\vartheta(u(z^{(1)},\infty^2)\tau)}\left(\phi(z) -\phi(z)^{-1}\right)&\frac{\vartheta(u(z^{(2)},\infty^2) -\epsilon;\tau)}{\vartheta(u(z^{(2)},\infty^2);\tau)}\left(\phi(z) +\phi(z)^{-1}\right)
      \end{pmatrix}
    \end{equation}
    where
    \begin{equation*}
        \phi(z)= \left(\frac{(z+i\alpha_{1})(z -i\alpha_{2})}{(z+i\alpha_2)(z -i\alpha_{1})}\right)^{\frac{1}{4}}\qquad \epsilon= \frac{x\Omega -i\Delta}{2\pi}
      \end{equation*}
and we referred with $z^{(1)},z^{(2)}$ the point $(z,\pm w)\in \mathcal{C}$ in the first and second sheet.
   
 {\color{black}Such solution is well defined. Indeed, the Jacobi elliptic function $\vartheta(z;\tau)$ has only one zero located in $z =\frac{1}{2}+\frac{\tau}{2}$ . Since  $\Omega$ is real  while $\Delta$ is imaginary, it follows that $\vartheta(\epsilon;\tau)\neq 0$ for all $x\in\R$.}
{\color{black} Furthermore, as  it is explained in \cite{DKZ},the ratio  $\frac{\phi(z) +\phi(z)^{-1}}{\vartheta(u(z^{(1)},\infty^1);\tau)}$ in the $11$ entry does not have poles but only fourth root singularities at the points $\pm i\alpha_1$ and $\pm i\alpha_2$. }The same considerations apply to the other  entries of the matrix $X(z;x)$.

{\color{black} Since $\phi(z) -\phi(z)^{-1}\to 0$ and $u(z^{(2)},\infty^1) \to -\frac{1}{2}$ as  $z \to \infty$, it is immediate to verify that $X(z)\to \1 + \mathcal{O}(z^{-1})$ as $z \to \infty$.}

\vskip 0.2cm
\noindent
{\bf The error parametrix near  the endpoints $\pm i\alpha_1,\pm i\alpha_2$.} 
The last step of the nonlinear steepest descent analysis is the definition and study of the error matrices around the end points of the segments $\mathcal{I}$ and $\overline{\mathcal{I}}$. Before analyzing the error parametrix, we should consider some assumptions about the behaviours of ${r} (z)$ near the points $\pm i\alpha_1, \pm i \alpha_2$. Indeed, Girotti et al. in~\cite{GGJMM} studied a RHP similar to~\eqref{eq:ell_z} for the mKdV and they proved that if the function ${r} (z)$ has a local behaviour near the end points $\pm i\alpha_1, \pm i \alpha_2$ of the form ${r} (z) \sim |z \pm i\alpha_j|^{\pm 1/2}Q(z)$ for $j=1,2$, with $Q(z)$ an analytic function locally bounded and non-zero in a neighbourhood of the end points, then it is possible to modify the lenses of the opening factorization so that the error matrices tends to the identity exponentially fast uniformly in $z\in \C$.   Specifically, we consider the following assumption.
 \begin{assumption}
     \label{ass:ass_1}
      Let $h>0$ and let us consider the open set   $\mathscr{U}_{h,+}$         defined as
     \begin{align}
       &\mathscr{U}_{h,+}:=\left\{z\in \mathbb{C} | \re(z)\in \left( 0\right. , \left. h\right]\right. \nn\\
       & \qquad \left.\text{ and } \alpha_1-\sqrt{h^2 -\re(z)^2}\le \im(z) \le \alpha_2 + \sqrt{h^2 - \re(z)^2} \right\},
     \end{align}
     with some $0 < h <\alpha_1$. The  open set $\mathscr{U}_{h,-}$ is defined by symmetry, $\mathscr{U}_{h,-}= \{z| - \ov{z} \in \mathscr{U}_{h,+}\}$. 

 When  ${r} (z)|z-i\alpha_j|^{\pm 1/2}$ is bounded and non zero on $\mathcal{I}$, we assume that ${r} (z)$ admits an analytical continuation to $\mathscr{U}_{h,-}\cup\mathscr{U}_{h,+}$
     \begin{align}
       &\hat{r}(z) \text{ analytic in } \mathscr{U}_{h,-}\cup\mathscr{U}_{h,+}  & \hat{r}(z)|_{z\in\mathcal{I}}={r} (z), & \\
       & \hat{r}_{+}(z) +\hat{r}_{-}(z)=0   & z \in [i\alpha_2,i (\alpha_2+ h)] \cup [i(\alpha_1 - h), i\alpha_1]. & 
     \end{align}
       \end{assumption}
   In our case, we have that ${r} (z)= \delta S(z)\beta^2(z)$, with $\beta^2(z)$ bounded in the original domain $\mathcal{D}$ and $\delta S(z)$ defined in $\mathcal{I}$ with behaviour at the end point of the form $\delta S(z) \sim |z-i\alpha_j|^{1/2}$. So, we are in the main hypothesis of the Assumption~\ref{ass:ass_1} and we consider $\mathscr{U}_{h,+}$ and $\mathscr{U}_{h,-}$ as the new lenses (see Figure~\ref{fig:3}).

   From this assumption, after we apply the transformation~\eqref{eq:S_transf}, we have jumps in the segments $[i\alpha_2,i(\alpha_2 + h)]$ and $[i(\alpha_1 -h),i\alpha_1]$    \begin{itemize}
   \item $z\in [i\alpha_2,i(\alpha_2 + h)]$
     \begin{equation}
       (T_{-}^{(2)}(z))^{-1}T_{+}^{(2)}(z)=
       \begin{pmatrix}
         1 & (\hat{r}_{+}(z)^{-1} +\hat{r}_{-}(z)^{-1})\frac{\e^{-2i(g(z) +z)}}{(f(z))^{2}}\\
           0 & 1
       \end{pmatrix}= {\bf 1};
     \end{equation}
   \item $z \in [i(\alpha_1 -h),i\alpha_1]$
        \begin{equation}
       (T_{-}^{(2)}(z))^{-1}T_{+}^{(2)}(z)=
       \begin{pmatrix}
         \e^{ix\Omega + \Delta} & (\hat{r}_{+}(z)^{-1} +\hat{r}_{-}(z)^{-1})\frac{\e^{-2i(g(z) +z)}}{(f(z))^{2}}\\
           0 & \e^{-ix\Omega -\Delta}
       \end{pmatrix}= \e^{(ix\Omega +\Delta)\sigma_3}.
     \end{equation}
   \end{itemize}
   
   This means that  we can   detach the lenses  from the points $i\alpha_1$ and $i\alpha_2$ so that they  full enclose the interval $\mathcal{I}$. The same applies to $\hat{r}^*(z)$ and with the complex conjugate $\overline{\mathscr{U}_{h+}}$ and $\overline{\mathscr{U}_{h-}}$ that enclose $\overline{\mathcal{I}}$.
   \begin{figure}[h!]
     \centering
       \includegraphics[width=0.7\textwidth]{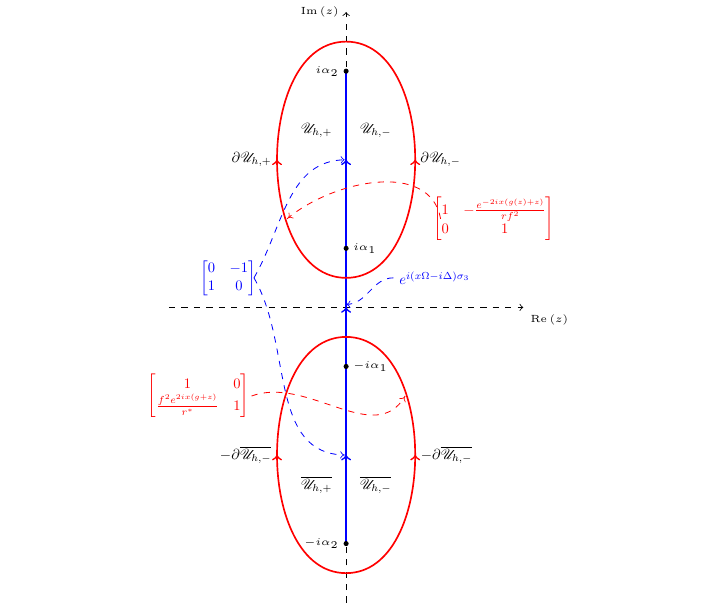}
     \caption{On the left: the modified new lenses $\mathscr{U}_{h,\pm}$. On the right: How the RHP~\eqref{eq:fin-rhp} change with the new lenses.}\label{fig:3}
   \end{figure}

   We now define the Error matrix
   \begin{equation}
     \mathcal{E}(z):= T^{(2)}(z)(X(z))^{-1},
   \end{equation}
   which is analytic for $z \in \C \setminus \pa(\mathscr{U}_{h,+} \cup \mathscr{U}_{h,-} \cup \overline{\mathscr{U}_{h,+}}\cup \overline{\mathscr{U}_{h,-}} )$ and it has the jump condition
   \begin{align}
     \mathcal{E}_{+}(z)= \mathcal{E}_{-}(z)\left\{
     \begin{matrix}
       \1 +\frac{ e^{-2ix(g(z) +z)}}{\hat{r}(z)f^{2}(z)} X(z)\sigma_+(X(z))^{-1} & \text{for } z \in \pa\mathscr{U}_{h+};\\
       \1 +\frac{ e^{-2ix(g(z) +z)}}{\hat{r}(z)f^{2}(z)} X(z)\sigma_+(X(z))^{-1} & \text{for } z \in \pa\mathscr{U}_{h-};\\
       \1 + \frac{ e^{2ix(g(z) +z)}f^{2}(z)}{\hat{r}^*(z)} X(z)\sigma_-(X(z))^{-1} & \text{for } z \in \pa\overline{\mathscr{U}_{h+}} ; \\
       \1 + \frac{ e^{2ix(g(z) +z)}f^{2}(z)}{\hat{r}^*(z)} X(z)\sigma_-(X(z))^{-1} & \text{for } z \in \pa\overline{\mathscr{U}_{h-}};
     \end{matrix}\right. 
   \end{align}
where the matrix $\sigma_+=\begin{pmatrix}0&1\\0&0\end{pmatrix}$ and  $\sigma_-=\begin{pmatrix}0&0\\1&0\end{pmatrix}$.
   Since $X(z)$ is bounded in $x$, the jump matrices tends to the identity exponentially fast with respect to the matrix norm.
%
   
   From the small norm lemma (see e.g. \cite{DKMVZ}), {we have that the matrix $\mathcal{E}$, uniformly in $ z \in \C$, tends exponentially to the identity as $x \to -\infty$}, i.e.
   \begin{equation}
   \label{E_exp}
     \mathcal{E}(z)= \1 + \mathcal{O}(\e^{c_-x}), \text{ as } x\to -\infty,
   \end{equation}
   with $c_- > 0$.
It follows that the model problem $X(z;x)$  coincides with   $T^{(2)}(z;x)$   up to an exponentially small error as $x\to-\infty$.
\vskip 0.2cm
\noindent
   {\bf Asymptotic  behaviour  of $\psi_0(x)$ as  $x \to \pm \infty$ and proof of Theorem~\ref{thmPsi0}}\\
   For $x \to +\infty$, the jump matrices of the RHP~\eqref{eq:ell_z} tend to the identity matrix exponentially fast, i.e. $(T(z;x))_{12} \sim \e^{-c_+ x}$, with $c_+ > 0$. From the equation~\eqref{Nsoliton} we have that
  \begin{equation*}
    \psi_{0}(x)= \mathcal{O}(\e^{-c_+ x}).
  \end{equation*}

  We now focus on the asymptotic behaviour for  $x \to -\infty$.
From the knowledge of $X(z;x)$ (exponentially close to the model problem $T^{(2)}(z)$)   we can find out the asymptotic behaviour of the initial datum $\psi_{0}(x)$ for the NLS equation.
  Indeed from \eqref{psi0T2}  and \eqref{E_exp}  we have 
\begin{equation}
\label{final_exp}
\begin{split}
\psi_0(x)&=2i\e^{i(g_{\infty}x+\phi_\infty)}\lim_{z\to\infty} \left[zT^{(2)}_{12}(z;x)e^{ig(z)x}f(z)\right]\\
&=2i\e^{i(g_{\infty}x+\phi_\infty)}\lim_{z\to\infty} \left[z(\mathcal{E}(z;x,t)X(z;x))_{12}e^{ig(z)x}f(z)\right]\\
&=2i\e^{i(g_{\infty}x+\phi_\infty)}\lim_{z\to\infty} \left[zX_{12}(z;x)e^{ig(z)x}f(z)\right] +\mathcal{O}(\e^{c_- x}), \text{ for } x \to -\infty\\
\end{split}
\end{equation}
From   \eqref{eq:g_R-H}   we obtain  as $z\to\infty$
   \begin{equation*}
     \e^{ig(z)x} =  \e^{ig_{\infty}x} (1 + \mathcal{O}\left(z^{-1}\right) ),\quad 
   \end{equation*}
     while for $f(z)$:
   \begin{equation*}
     f(z)= e^{i\phi_{\infty}}+ \mathcal{O}\left(z^{-1}\right).
   \end{equation*}
Next we need to expand $X_{12}(z)$:
  \begin{equation*}
    \begin{split}
      X_{12}(z)&=-i\frac{\vartheta(0)\vartheta(u(z^{(2)},\infty_1) -\epsilon)}{2\vartheta(\epsilon)\vartheta(u(z^{(2)},\infty_1))}\left(\phi(z) -\phi(z)^{-1}\right) \\
      &= -i\frac{\vartheta(0)\vartheta(\epsilon+ \frac{1}{2})}{2\vartheta(\epsilon)\vartheta(\frac{1}{2})}\left(-i\frac{(\alpha_{2}-\alpha_1)}{z} + \mathcal{O}\left(z^{-2}\right)\right)\\
      &= -\frac{\vartheta(0)\vartheta(\epsilon+\frac{1}{2})(\alpha_2 - \alpha_1)}{2\vartheta(\epsilon)\vartheta( \frac{1}{2})z} + \mathcal{O}\left(z^{-2}\right)
     \end{split}
   \end{equation*}
 where $u(\infty^2,\infty^1)=-\frac{1}{2}$ by the symmetry of the problem.  Inserting the above three expansions into \eqref{final_exp} we obtain \eqref{eq:NLS-t}.
\hfill$\square$

\begin{remark}
  The formula for the elliptic solution \eqref{eq:NLS-t} can be rewritten in terms of the Jacobi elliptic function  $\dn(z;m)$.
  Let us introduce the   Jacobi theta-functions
  \begin{align}
 & \theta_{3}(z;\tau) := \vartheta(z;\tau), \\
  & \theta_{4}(z;\tau) := \vartheta\left(z + \frac{1}{2};\tau\right)
\end{align}
  and  Jacobi elliptic functions 
  \begin{align}
  \label{dn}
  &\dn(2K(m)z;m)=\dfrac{\theta_4(0;\tau)}{\theta_3(0;\tau)}\dfrac{\theta_3(z;\tau)}{\theta_4(z;\tau)}
  \end{align}
%
  Then, using the above relation,  we can rewrite the initial datum as
\begin{gather}
	\begin{aligned} \label{q.bg1}
	\psi_{0}(x)
	&=-i\e^{2i(g_{\infty}x+\phi_\infty)} (\alpha_2 -\alpha_1) 
	\frac{ \theta_3(0; \tau )}{\theta_3(\frac{1}{2}; \tau )} 
	  \frac{\theta_3( \tfrac{1}{2} +\frac{\Omega x  - i\Delta}{2 \pi}; \tau)}{ \theta_3(\frac{\Omega x  - i\Delta}{2 \pi}; \tau)} \\
&=-i \e^{2i(g_{\infty}x+\phi_\infty)}(\alpha_2 -\alpha_1)  \dfrac{1}{\dn \left( \frac{K(m)}{\pi} ({\Omega}x-i\Delta);  m \right)}	\\
	&=-i \e^{2i(g_{\infty}x+\phi_\infty)}\frac{\alpha_2 -\alpha_1}{\sqrt{1-m}	}	\dn \left( \frac{K(m)}{\pi} ({\Omega}x-i\Delta +\pi);  m \right)\\	 &=-i \e^{2i(g_{\infty}x+\phi_\infty)}(\alpha_2 +\alpha_1) 
		\dn \left(  (\alpha_2 +\alpha_1)  ( x -x_0); m \right),	
	\end{aligned}
	\end{gather}
%
%
where  in the second  relation we used the identity $\dn(u+K(m);m)=\dfrac{\sqrt{1-m}}{\dn(u;m)}$ and in the last relation we use the parity of $\dn(-u;m)=\dn(u;m)$ and plug in the explicit values of $\Omega$ and $\Delta$ as in \eqref{eq:Omega} and \eqref{eq:delta} and  $$x_0=\frac{1}{2\pi}\left[\int_{i\alpha_1}^{i\alpha_2}\frac{\log {r} (\zeta)}{\sqrt{R_{+}(\zeta)}}d \zeta - \int_{-i\alpha_2}^{-i\alpha_1}\frac{\log {r} ^{*}(\zeta)}{\sqrt{R_{+}(\zeta)}}d \zeta \right]-\frac{K(m)}{\alpha_1+\alpha_2}$$
In the limit $m\to 1$,  or $\alpha_2\to\alpha_1$,  we have  that $\dn(u;m)\sim \mbox{sech}(u)$ and $K(m)\sim\frac{1}{2}\log\frac{8}{(1-m)}$ so that  $\psi_{0}(x)$ tends to the initial datum of the one soliton solution~\eqref{one_soliton}, with $x_0 \to -\infty$.
In the limit $\alpha_1\to0$ we have $m\to 0$ and the elliptic function $\dn(z,m)\to 1$  so that the NLS initial datum $\psi_0(x)$ is a plane wave at $x\to-\infty$.
\end{remark}

\begin{remark}
The initial data described by Theorem~\ref{thmPsi0}  is step-like oscillatory. The  long time asymptotic behaviour of the solution of  the  NLS  equation   is being considered in \cite{GMO} and it is inspired by  the asymptotic analysis performed  for step-like plane wave initial data  considered in  \cite{BiondiniMan}, \cite{BoutetdeMonvel2021},\cite{Monvel2022a}
\end{remark}
\appendix

\noindent {\bf Acknowledgements}

\noindent
 TG  acknowledge the support of PRIN 2022 (2022TEB52W)  "The charm of integrability: from nonlinear waves to random matrices"-– Next Generation EU grant – PNRR Investimento M.4C.2.1.1 - CUP: G53D23001880006; the GNFM-INDAM group and the  research project Mathematical Methods in NonLinear Physics (MMNLP), Gruppo 4-Fisica Teorica of INFN.  
 The work of MB was supported in part by the Natural Sciences and Engineering Research Council of Canada (NSERC) grant RGPIN-2023-04747.
 
 \noindent
 GO acknowledge the support of FNRS Research Project T.0028.23.
 
\noindent 
The authors would like to thank the Isaac Newton Institute for Mathematical Sciences, Cambridge, for support and hospitality during the programme "Emergent phenomena in nonlinear dispersive waves", where work on this paper was undertaken. This work was supported by EPSRC grant EP/R014604/1.

\appendix
\section{Derivation of \eqref{eq:tau-N}}
\label{IIKStau}
According to the general IIKS framework, the  jump of  RHP \eqref{eq:new_R-H0}
can be written in the form 
\[
J(z)=\1+2\pi i p(z)q^T(z)
\]
with 
\begin{align}
p(z)&=\begin{pmatrix}
 {\bf 1}_{\gamma_-}(z)\\
{\bf 1}_{\gamma_+}(z)\end{pmatrix},
\;\;q(z)=\begin{pmatrix}
-A (z) {\bf 1}_{\gamma_+}(z)
\nn \\
A^\star(z){\bf 1}_{\gamma_-}(z)\end{pmatrix},
\\
A(z) &:=  \sum_{j=1}^N \frac{c_j {\rm e}^{2\theta(z)}}{2i\pi(z-z_j)} ,\quad A^\star (z) :=  \sum_{j=1}^N \frac{\ov{c_j} {\rm e}^{-2\theta(z)}}{2i\pi(z-\ov{z_j})} \,.
\nn
\end{align}
The corresponding integrable operator   $\mathcal{T}$ associated to RHP \eqref{eq:new_R-H0}  acts on $L^2(\gamma_+\cup \gamma_-) \simeq L^2(\gamma_+) \oplus L^2(\gamma_-)$
and is defined by  the  kernel
\bea
  \hat{\mathcal{T}}(z,w):&= \frac{p^T(z) \, q(w)}{z - w}\\
  &= -\frac {A(w) {\bf 1}_{\gamma_+}(w) {\bf 1}_{\gamma_-}(z) - A^\star(w) {\bf 1}_{\gamma_-}(w) {\bf 1}_{\gamma_+}(z)}{z-w}.\nn
\eea
The corresponding $\tau$ function is 
\[
\tau=\det(\id_{L^2(\gamma_+\cup \gamma_-) }-\mathcal{T}).
\]
%
%
In terms of the splitting $L^2(\gamma_+\cup \gamma_-) \simeq L^2(\gamma_+) \oplus L^2(\gamma_-)$ the Fredholm determinant to compute is 
\be
\tau = \det \left[
\begin{array}{cc}
\id_{L^2(\gamma_+)}  & \mathcal{T}_{-+}\\
\mathcal{T}_{+-}  &\id_{ L^2(\gamma_-)}
\end{array}
\right]
\ee
where $\mathcal{T}_{-+} :L^2(\gamma_+) \to L^2(\gamma_-)$ and $\mathcal{T}_{+-} :L^2(\gamma_-) \to L^2(\gamma_+)$ are given by 
\be
\mathcal{T}_{-+} [\phi](z) = {\bf 1}_{\gamma_-}(z)\oint_{\gamma_+} \frac {\phi(w) A(w) d w }{z-w}
\nn
\\
\mathcal{T}_{+-} [\psi](z) = {\bf 1}_{\gamma_+}(z)\oint_{\gamma_-} \frac {\psi(w) A^\star(w) d w }{w-z}
\ee
By the standard determinantal identities we have 
\be
\label{eq:tau-app}
\tau = \text{\rm det}_{L^{2}(\gamma_{+})}\left[\id - \mathcal{T}_{+-}\circ \mathcal{T}_{-+}\right]
\ee
One can see that the operator is of finite rank. Indeed
\be
\mathcal{T}_{+-}\circ \mathcal{T}_{-+}[\phi](z) =\oint_{\gamma_-} \frac { A^\star(s) d s }{s-z} \oint_{\gamma_+} \frac {\phi(w) A(w) d w }{s-w}
\ee
and the $s$--integration results in a residue evaluation and the poles of $A^\star$ at $\ov z_j$'s:
\bea
\mathcal{T}_{+-}\circ \mathcal{T}_{-+}[\phi](z)&= \sum_{j=1}^N \frac {\ov {c_j} {\rm e}^{-2\theta(\ov z_j)}}{\ov z_j-z}  \oint_{\gamma_+} \frac {\phi(w) A(w) \d w }{\ov{z_j}-w}\nn\\
&= \sum_{j=1}^{N} \widetilde{A}(v_{j},\phi) v_{j}(z),
\eea
where $v_{k}(z) \in L^{2}(\gamma_{+})$, $\widetilde{A}$ is a bilinear form in $L^{2}(\gamma_{+})$ and they are defined as
\bea
&\widetilde{A}(\varphi,\psi):= \oint_{\gamma_{+}} \varphi(w)A(w)\psi(w)\d w &\varphi(z),\psi(z) \in L^{2}(\gamma_{+}); \\
&v_{k}(z):= \frac{\sqrt{\ov c_{k}}\e^{-\theta(z_{k})}}{\ov{z_{k}} - z} & \text{ for } k= 1;
,\dots, N.
\eea
Therefore the determinant~\eqref{eq:tau-app} becomes a finite dimensional Fredholm determinant of the Gram matrix, namely
\bea
 \widetilde{A}(v_{j},v_{k})&= \oint_{\gamma_+}  \frac {\sqrt{\ov {c_j}\ov{c_k}} {\rm e}^{-\theta(\ov z_j)-\theta(\ov z_k)}}{\ov z_j-z}
 \frac{A(z)}{\ov z_k-z}  d z \nn\\
 &=\sum_{\ell=1}^N   \frac{\sqrt{\ov {c_j}\ov{c_k}} {\rm e}^{-\theta(\ov z_j)-\theta(\ov z_k)}}{\ov z_j-z_\ell} \frac {{c_\ell} {\rm e}^{2\theta(z_\ell )}}{ \ov{z_k}-z_\ell}\nn \\
& = -\big(\ov{\Phi_N} \Phi_N\big)_{jk}
\eea
where $\Phi_{N}(x,t)$ is defined in~\eqref{Phi}.

Then, up to a conjugation, we have rewritten the Fredholm determinant~\eqref{eq:tau-N}.

\bibliography{Solitonellipse}{}
\bibliographystyle{siam}

\end{document}